\documentclass[10pt,twocolumn]{article}
\usepackage[T1]{fontenc}
\usepackage[utf8]{inputenc}
\usepackage{lmodern}
\usepackage[letterpaper,margin=0.75in,columnsep=0.25in]{geometry}
\usepackage{amsmath,amssymb,amsthm}
\usepackage{booktabs,tabularx,array,microtype}
\usepackage[numbers,sort&compress]{natbib}
\usepackage[hidelinks]{hyperref}
\newtheorem{theorem}{Theorem}[section]
\newtheorem{lemma}[theorem]{Lemma}
\newtheorem{corollary}[theorem]{Corollary}
\newtheorem{proposition}[theorem]{Proposition}
\theoremstyle{definition}
\newtheorem{definition}[theorem]{Definition}

\theoremstyle{remark}

\allowdisplaybreaks[1]
\title{Conditional Indexing Lower Bounds and Normalization Boundaries\protect\\for Elastic Subsequence Matching}
\author{Zemin Chao \quad Zitong Li \quad Xiangyu Gao \quad Hongzhi Wang\thanks{Corresponding author.}}
\date{August 2026}
\hypersetup{pdftitle={Conditional Indexing Lower Bounds and Normalization Boundaries for Elastic Subsequence Matching},pdfauthor={Zemin Chao, Zitong Li, Xiangyu Gao, Hongzhi Wang}}
\begin{document}
\maketitle
\begin{abstract}
We study static indexes for exact subsequence matching over all equal-length windows of a time series. Under deterministic SETH, for every fixed $\beta\ge1$ and $\delta\in(0,1)$, we exclude $O(N^\beta)$ preprocessing together with $O(N^{1-\delta})$ worst-case query time on domains containing the reduction's polylogarithmic query lengths. The bounded-error counterpart assumes randomized SETH. The results cover specified warping, bottleneck, edit, subsequence, local-alignment, and complexity-adjusted scores. Their reductions share a data-independent query interface but require different arguments for all-window soundness. We give complete representative proofs for the distinct mechanisms and state the changes needed by related scores. For soft-DTW, integer scaling yields an exact Boolean-predicate reduction in Word-RAM; a model with exact transcendental primitives requires a separate assumption. For independent window-wise Z-normalization, balanced coordinate gadgets and shifted-window moment bounds yield exact lower bounds for squared-cost DTW and discrete Fr\'echet distance. We give finite-word comparisons and a specified length-dependent rounding rule preserving these gaps. Raw and normalized discrete Fr\'echet constructions also exclude feasible-solution approximation within any fixed factor below two. Finally, we prove conditional signal-attenuation bounds for normalized anchor constructions and a fast index under the all-window common-mean promise.
\end{abstract}

\section{Introduction}
\label{in:introduction}

Subsequence matching asks which contiguous part of a long time series resembles a short query.  Elastic scores accommodate local misalignment, missing observations, and differences in sampling, using warping paths, edit operations, or matched subsequences \cite{berndt1994using,chen2004marriage,vlachos2002discovering}. The indexing question is whether preprocessing the target can guarantee a polynomial saving over linear query time on every legal input.  We study a static target $S$ of length $N$, an online query $Q$ of length $m$, and \emph{all} contiguous length-$m$ windows of $S$.  The decision predicate asks whether any such window has distance at most a supplied threshold, or similarity at least a threshold.  This candidate domain is part of the problem: windows may start at any position, including inside an encoding block or across two blocks.

We give reductions from Orthogonal Vectors Indexing for the measures in Table~\ref{in:catalog}. The target encoding depends only on the stored vectors, while the query encoding depends only on the online vector. Correctness must account for every physical window and every permitted alignment. This requires different arguments for different scores: edit operations may remove anchors, warping paths may repeat matches, subsequence alignments may skip symbols, and local alignments may restart. We bound these alternatives and give finite-word representations for the constructed inputs and thresholds.

\subsection{Main results}

\begin{theorem}[Main indexing statement]
\label{in:master}
Consider any exact Word-RAM score slice specified in Table~\ref{in:catalog}, with the definitions and parameter restrictions proved in the corresponding sections.  For every fixed $\beta\ge1$ and $\delta\in(0,1)$, deterministic SETH excludes a uniform exact static index with preprocessing time $O(N^\beta)$ and worst-case query time $O(N^{1-\delta})$ on a length domain covering the reduction's hard family.  The dimension coefficient is chosen \emph{after} $(\beta,\delta)$: $d=\lceil c\log_2n\rceil$ for a sufficiently large fixed $c$, yielding $N=nd^{O(1)}$ and $m=d^{O(1)}$.

Under randomized SETH, the same exclusion holds for every fixed per-input two-sided error $p<1/2$, allowing $c$ to depend also on $p$. Error is measured jointly over preprocessing and query randomness; both running-time bounds are worst case over inputs and random choices. The raw and independently normalized DFD constructions additionally exclude feasible-solution approximation with each fixed $1\le\rho<2$, under the corresponding deterministic or randomized assumption and output guarantee.
\end{theorem}
\begin{proof}
Each score section supplies its all-window equivalence, target/query separability, size bound, and numerical certificate. Lemmas~\ref{lem:rt-ovi} and~\ref{lem:rt-transfer} give the exact indexing exclusions.  The two DFD approximation results use their explicit distance-$1$ versus distance-$2$ gaps (up to a common normalization factor) and polynomial-in-$m$ verification of the returned feasible window or coupling.
\end{proof}

The statement concerns fixed, exactly implementable parameter slices; a uniform algorithm accepting parameters explicitly is also ruled out if its input domain contains a proved slice.  It does not establish an unconditional $\Omega(N)$ bound, a result for constant $m$, or a result for every prescribed small coefficient in $m=\Theta(\log N)$. The shorthand conditional $N^{1-o(1)}$ refers only to a uniform domain covering the hard lengths needed for the exponents in question.

\begin{table*}[tp]
\centering
\small
\setlength{\tabcolsep}{4pt}
\renewcommand{\arraystretch}{1.15}
\begin{tabularx}{\textwidth}{@{}p{0.17\textwidth}p{0.38\textwidth}X@{}}
\toprule
Score family & Score on one window $X$ and query $Q$ & Proved setting and consequence \\
\midrule
DTW / cDTW & $\min_{\pi\in\Pi}\sum_{(i,j)\in\pi}|x_i-q_j|$ & Global endpoints; unrestricted or fixed Sakoe--Chiba band; exact indexing lower bound \\
cWDTW & $\min_{\pi\in\Pi}\sum_{(i,j)\in\pi}\omega(|i-j|)|x_i-q_j|$ & Fixed exactly implementable weights with a uniform positive lower bound \\
Itakura-DTW & DTW formula with $\pi\in\Pi_m^{\rm Ita}$ & Fixed path rule containing the complete diagonal on the hard lengths \\
Discrete Fr\'echet & $\min_{\pi\in\Pi}\max_{(i,j)\in\pi}|x_i-q_j|$ & Exact indexing and feasible approximation with each fixed $1\le\rho<2$ \\
ERP & $\min_{\mathcal A}\sum_{(a,b)\in\mathcal A}|a-b|$ & Global edit columns with gap value $g=0$; fixed representable gaps by common translation \\
TWED & Minimum sum of the three transition costs in~\eqref{em:twed-dp} & Local unit timestamps, infinite boundaries, fixed $\nu,\lambda\ge0$ \\
MSM & Minimum Move and Split/Merge costs in~\eqref{em:msm-dp} & Fixed exactly implementable penalty $\lambda>0$ and the stated initialization \\
cDDTW & $\operatorname{DTW}(\mathcal D(X),\mathcal D(Q))$ & Endpoint-trimmed internal derivative; absolute local cost; fixed band or unrestricted paths \\
LCSS & $L(X,Q)=\max\{|\mathcal A|:\mathcal A\text{ is an ordered match}\}$ & Fixed integer-symbol tolerance slice, without an additional time band \\
EDR & Minimum number of substitutions, insertions, and deletions & Matches cost zero; other edits cost one; fixed tolerance slice \\
Swale & $W(X,Q)=(3L(X,Q)-|X|-|Q|)r$ & Fixed reward $+r$ and one-sided gap cost $-r$; global recurrence with no mismatch diagonal \\
TS-SW & $T(X,Q)=\max_{i,j}H_{i,j}$, with reward $4-|x_i-q_j|$ & Zero-restart local recurrence and unit gap penalty; exact lower bound with a positive vanishing gap \\
Raw CID & $\displaystyle\|X-Q\|_2\,\frac{\max\{\mathrm{CE}(X),\mathrm{CE}(Q)\}}{\min\{\mathrm{CE}(X),\mathrm{CE}(Q)\}}$ & Positive-complexity raw inputs; exact squared rational comparison \\
Soft-DTW & $-\gamma\ln\sum_{\pi\in\Pi}\exp(-C_\pi/\gamma)$ & Absolute path cost, fixed positive rational $\gamma$; exact Word-RAM predicate; separate extended-model assumption \\
Normalized squared DTW / DFD & $\begin{gathered}\min_\pi\sum_{(i,j)\in\pi}(Z_X(i)-Z_Q(j))^2\\\min_\pi\max_{(i,j)\in\pi}(Z_X(i)-Z_Q(j))^2\end{gathered}$ & Independent positive-variance normalization; exact finite-word predicates \\
Ordinary normalized DFD & Square root of normalized squared DFD & Exact rational-threshold indexing and each fixed feasible approximation factor $\rho<2$ \\
Rounded normalized scores & Replace each local square $a$ by $\lfloor R(m)a\rfloor/R(m)$ & Prescribed dyadic $R(m)\ge8m^5$; subsequent DP operations are exact integer operations \\
\bottomrule
\end{tabularx}
\caption{Score definitions and scope of the indexing results. The shared exponent, length-domain and randomness conventions are in Theorem~\ref{in:master}. Symbols $\mathcal D$, $\mathrm{CE}$, $H$, and the edit recurrences are defined in their score sections. The normalized construction uses squared local costs; the raw DTW rows use absolute costs. Table~\ref{tab:gap-catalog} gives the numerical separations.}
\label{in:catalog}
\end{table*}

The constructions start from a bounded asymmetric embedding of Boolean vectors. For each measure, we use its recurrence to rule out alignments that could invalidate the encoded distance gap. The arguments use direct anchor matches, cumulative costs, projection onto complete cores, or mandatory endpoint costs. Derivative DTW requires an inverse construction and a bound on derivatives across block boundaries \cite{keogh2001derivative}. For CID, we also control the multiplicative complexity factor~\cite{batista2011cid}.

Soft-DTW requires an additional entropy bound. Scaling the integer input costs preserves a strict gap after soft aggregation, giving an exact decision reduction on the Word-RAM. The reduction does not evaluate exponentials. The model with unit-cost exact transcendental operations has a separate OnlineOV assumption. For approximate score evaluation, we state the uniform error bound needed to preserve the decision gap.

Independent window normalization is handled by a different embedding. Zero-sum, equal-energy coordinate gadgets give common moments to the canonical blocks, while explicit perturbation bounds control all shifted windows. The cost at the first cell then excludes those windows for squared-cost DTW and DFD. We give exact comparisons of the resulting quadratic-root expressions and a length-dependent dyadic rounding rule that preserves the gap. We also examine when normalization attenuates amplitude-based constructions. In particular, requiring a common mean for all windows forces periodicity and leaves at most $m$ distinct windows; the normalized reduction only requires equal moments on canonical blocks.

\paragraph{Organization of the proofs.}
The common complexity argument is proved once in Section~\ref{sec:rt-model}. The remaining proofs are organized by the mechanism that excludes an unintended witness. Section~\ref{sec:rt-raw} gives the complete first-departure argument for additive DTW; bottleneck aggregation, positive weights, and fixed path restrictions reuse it with explicit changes to the threshold. Section~\ref{em:section} treats actual edit transitions, with ERP as the complete alignment-column argument and separate sketches for the variation and Split/Merge costs. Section~\ref{ds:section} proves the repeated-anchor and buffered-core arguments before transferring the latter to Swale. The inverse-derivative boundary calculation, the soft-min gap, and the independently normalized moment estimates have distinct premises and are proved separately. A proof sketch specifies the inherited lemma and the changed inequalities; all statements retain their parameter and output restrictions.

The constructions and their quantitative gaps are collected in Table~\ref{tab:gap-catalog}. In addition to the indexing exclusions, Section~\ref{fw:normalization} establishes explicit perturbation bounds for absolute-cost DTW, soft-DTW, and DFD under a shared anchor skeleton. It also gives an $O(N)$-space index with $O(N^2)$ preprocessing and polynomial-in-$m$ query time on the all-window common-mean promise. These auxiliary results delimit the normalization argument and are separate from the positive normalized lower bounds.

\subsection{Related work}

\paragraph{Subsequence search and preprocessing.}
Feature-space subsequence indexing~\cite{faloutsos1994fast}, exact DTW indexes~\cite{keogh2002exact,niennattrakul2010exact}, and large-scale search with pruning and early abandonment~\cite{rakthanmanon2012searching} reduce practical search costs.  KV-Match addresses normalization and time warping~\cite{wu2019kvmatch}; FSMDTW reduces the average overhead of candidate-window lower-bound evaluation~\cite{chao2025fsmdtw}. Stronger safe distance bounds improve filtering \cite{lemire2009faster,webb2021tight}.  These numerical lower bounds support algorithms; our conditional \emph{running-time} lower bounds address a different question and do not preclude those improvements.

A direct predecessor is Chao et al.~\cite{chao2022inherent}, who give conditional query lower bounds for raw $L_p$ subsequence matching after polynomial preprocessing and connect a zero-bandwidth slice to raw DTW. The present reductions treat the specified elastic paths, all shifted equal-length windows, fixed parameter restrictions, and independent normalization. Goldenberg, Rubinstein, and Saha~\cite{goldenberg2020preprocessing} study comparison after separately preprocessing two strings, with related hardness under BP-SETH.  That result has a different query object and hypothesis from our static-text, all-window SETH formulation.

\paragraph{Sequence complexity and low-distance DTW.}
Conditional pairwise lower bounds for edit distance, DTW, LCS, and Fr\'echet distance underpin fine-grained sequence complexity \cite{backurs2015edit,bringmann2015quadratic,abboud2015tight,bringmann2014why}. Faster low-distance or approximate algorithms \cite{kuszmaul2019dynamic}, comparison to a fixed sequence \cite{abboud2021fixed}, and dynamic DTW data structures \cite{bringmann2024dynamic} make input and resource restrictions particularly important when comparing bounds.

Boneh, Golan, and Kociumaka~\cite{boneh2026low} give low-distance DTW and pattern-matching reductions.  We observe that specializing the query vector family to a singleton in their Section~5.1 construction makes the text depend only on the data family and dimension, while the pattern depends only on the online vector and has length $O(d^2)$.  Combining this separation with OnlineOV gives fixed-text indexing hardness for their setting. Their seven-symbol unit-mismatch cost and arbitrary nonempty substring domain differ from the numerical absolute costs and equal-length windows considered here. The latter require separate soundness arguments.

\paragraph{Fr\'echet queries and proximity structures.}
Blank~\cite{blank2026imbalanced} establishes one-dimensional discrete Fr\'echet whole-curve query hardness after polynomial preprocessing, including a below-two approximation barrier.  Gudmundsson et al. \cite{gudmundsson2024mapmatching} prove preprocessing barriers for two-dimensional continuous Fr\'echet whole-trajectory queries. Bringmann et al.~\cite{bringmann2022frechetann} study preprocessing--query bounds for continuous Fr\'echet near-neighbor search over curve sets. Whole curves, a curve collection, and all windows of one text have different candidate domains: discarding text prefixes and suffixes can create a match, so whole-curve soundness does not establish all-window soundness.  The general themes of preprocessing hardness and a factor-two DFD barrier are consequently established background.

Curve hashing and deterministic approximate near-neighbor structures \cite{driemel2017lsh,filtser2020approximate} use approximation and curve-complexity restrictions.  Aronov et al.~\cite{aronov2020proximity} consider DTW proximity problems with segment queries, again a distinct query domain. These results concern different candidate families and approximation guarantees from those in Table~\ref{in:catalog}. OnlineOV itself is prior work~\cite{goldstein2017data,gajulapalli2026online}; our self-contained derivation records the exponent order and randomized interface needed for the reductions.

Sections~\ref{sec:rt-model} and~\ref{fw:framework} give the computational model and reduction lemmas. Sections~\ref{sec:rt-raw}--\ref{ds:section} treat the raw scores, and Section~\ref{sec:rt-soft} treats soft-DTW. Sections~\ref{fw:normalization} and~\ref{sec:z} address normalization.

\section{Preliminaries}
\label{sec:rt-model}

\subsection{Static subsequence queries}

For a sequence $S$ of length $N$, let $\operatorname{Win}_m(S)$ contain all $N-m+1$ consecutive length-$m$ windows. Preprocessing receives only $S$ and constructs an index $\mathcal I(S)$. Subsequently a query supplies a sequence $Q$, a threshold $\tau$, and any explicit score parameters. A query reads the stored index without modifying it. For a finite-valued minimizing score $D$, write
\begin{align}
 D_{\rm sub}(S,Q)&=\min_{X\in\operatorname{Win}_{|Q|}(S)}D(X,Q),\nonumber\\
 \text{answer Yes}\quad&\Longleftrightarrow\quad D_{\rm sub}(S,Q)\le\tau.
 \label{eq:rt-query}
\end{align}
For a maximizing similarity, use the maximum and reverse the threshold inequality. Preprocessing is independent of $Q$ and $\tau$. Window starts range over every physical position, including those inside an encoding block or crossing a boundary.

\begin{definition}[Query-length domain]\label{def:query-domain}
A query-length domain is a prescribed family $\mathcal L=(\mathcal L_N)_{N\ge1}$ with $\mathcal L_N\subseteq\{1,\ldots,N\}$. Correctness and the stated query-time bound must hold uniformly for every legal query with $|Q|\in\mathcal L_N$. Supporting all lengths means $\mathcal L_N=\{1,\ldots,N\}$; a single size-dependent slice may instead have $\mathcal L_N=\{m(N)\}$. Equal-length windows are defined relative to the current query and do not require $m(N)=O(1)$.
\end{definition}

A parameter fixed in a theorem is independent of $n$ and $d$; the reduction's constants may depend on it. A uniform parameterized algorithm accepts the parameter explicitly and must handle all legal values. Restricting that algorithm to a proved fixed slice is therefore valid. Query time includes reading the query, threshold, and explicit parameters, as well as producing the required output. The hard queries below have $m=\operatorname{polylog}N$, so their input-reading cost is $N^{o(1)}$.

\subsection{Computation and complexity hypotheses}

We use a uniform Word-RAM with word length $w=\Theta(\log N)$. The same finite algorithms apply at all input sizes. All hard inputs have explicitly encoded integers or rationals occupying $O(1)$ words, and every required primitive is implemented by $O(1)$ word operations. Most encodings have polynomial magnitude in $d=\Theta(\log N)$; inverse-derivative reconstruction has $2^{O(d)}$ magnitudes and still uses $O(\log N)$ bits. Section~\ref{fw:arithmetic} gives the numerical criterion. The separately specified extended Real-RAM for soft-DTW is not part of this default model.

Preprocessing and query times are worst case over legal inputs and, when applicable, all random choices. For each fixed target and query, correctness probability is taken jointly over preprocessing and query randomness. A fixed error bound $p<1/2$ concerns that input; it does not promise simultaneous correctness of all possible queries. The SAT reductions below use a fixed query collection, and amplification supplies the batch guarantee they require.

Deterministic SETH asserts that for every $\varepsilon>0$ some fixed $k$-SAT has no $O^*(2^{(1-\varepsilon)v})$ deterministic algorithm on $v$ variables~\cite{impagliazzo2001ksat,impagliazzo2001complexity}. Randomized SETH makes the same assertion for algorithms with error at most $1/3$ on each input and worst-case running time. We use the deterministic assumption for deterministic indexes and the randomized assumption for randomized indexes. Randomized SETH appears in lattice hardness~\cite{aggarwal2018gapsvp,bennett2020bdd}; the preprocessing BDD result uses a separate nonuniform hypothesis, while the random-algebraic approximation work of Bukh et al.~\cite{bukh2022randomalgebraic} uses a one-sided Monte Carlo formulation. These conventions are distinct from the two-sided per-input assumption adopted here.

\subsection{The OnlineOV lower bound}

Online Orthogonal Vectors (OnlineOV) preprocesses an ordered family $A=(a^{(1)},\ldots,a^{(n)})$ of vectors in $\{0,1\}^d$ and answers, for an online vector $q\in\{0,1\}^d$, whether some $\langle a^{(i)},q\rangle=0$. Repeated records are allowed. Its preprocessing--query tradeoffs are established background \cite{goldstein2017data,gajulapalli2026online}; the quantifiers used here are recorded explicitly. All complexity logarithms are base two; soft-min formulas explicitly use natural logarithms.

\begin{lemma}[Amplification with preprocessing randomness]\label{lem:rt-amplification}
Fix $(A,q)$ and suppose its joint preprocessing-and-query error is at most $p<1/2$. Run an odd number $K$ of independent copies of the entire index and answer by majority. The error is at most
\[
 \exp\bigl(-2K(1/2-p)^2\bigr).
\]
For $M\ge2$ and fixed $\lambda>0$, it is at most $M^{-(\lambda+1)}$ when
\[
 K\ge\frac{(\lambda+1)\ln M}{2(1/2-p)^2}.
\]
Preprocessing time, stored space, and query time increase by at most a factor $K$.
\end{lemma}
\begin{proof}
Let $E_i$ indicate an incorrect answer in copy $i$. Independence includes the preprocessing coins, so $E_1,\ldots,E_K$ are independent and have expectations at most $p$. An incorrect majority requires $\sum_i E_i\ge K/2$. Hoeffding's inequality~\cite{hoeffding1963probability} gives the bound; the choice of $K$ gives the stated power of $M$. Constructing and querying $K$ copies gives the resource bounds. Repeating queries against only one randomized preprocessing outcome would not supply this independence.
\end{proof}

\begin{lemma}[Parameter-sensitive OnlineOV lower bound]
\label{lem:rt-ovi}
Fix $b\ge1$ and $\epsilon\in(0,1)$. Under deterministic SETH, there exists $c_0(b,\epsilon)$ such that, for every fixed $c\ge c_0$, no uniform OnlineOV index on $d=\lceil c\log_2 n\rceil$ dimensions has preprocessing $O(n^b)$ and query time $O(n^{1-\epsilon})$. Under randomized SETH the same holds for every fixed $p\in[0,1/2)$, allowing $c_0$ also to depend on $p$, with the joint per-input error and worst-case time conventions above.
\end{lemma}
\begin{proof}
We first prove the deterministic statement. Put $\theta=1/(2b)$ and $\eta=\theta\epsilon$. Deterministic SETH supplies a fixed clause width $k$ for which $k$-SAT on $v$ variables has no $O^*(2^{(1-\eta/2)v})$ algorithm. Sparsification with parameter $\eta/4$ produces at most $2^{\eta v/4}$ formulas, each with at most $s v$ clauses for a fixed $s$ depending on $k,\eta$~\cite{impagliazzo2001complexity}. Here $O^*$ suppresses only polynomial factors in $v$.

For one sparse formula, assign $\lfloor\theta v\rfloor$ variables on the data side and the remaining variables on the query side. A coordinate is one exactly when the partial assignment leaves the corresponding clause unsatisfied. A data vector and query vector are orthogonal if and only if their combined assignment satisfies the formula. Keep all assignment records, even when two records give the same vector. Thus
\[
 n=2^{\lfloor\theta v\rfloor},\qquad
 \#\text{queries}=2^{v-\lfloor\theta v\rfloor}.
\]
Choose a fixed $c_0$ large enough that $s v\le c_0\lfloor\theta v\rfloor$ for all sufficiently large $v$. For every fixed $c\ge c_0$, zero padding yields dimension $d=\lceil c\log_2 n\rceil$ without changing orthogonality. The same algorithm and padding rule apply at every sufficiently large input size.

The supposed index preprocesses in $O^*(2^{v/2})$ time and answers the whole query family in
\[
 O^*\!\left(2^{v-\theta v}\,2^{\theta(1-\epsilon)v}\right)
 =O^*(2^{(1-\eta)v}).
\]
Generating all query vectors takes $O^*(2^{(1-\theta)v})$ and is absorbed because $\eta<\theta$. Data generation is absorbed as well. Summing over the sparse formulas gives $O^*(2^{(1-3\eta/4)v})$, which is strictly faster than the excluded $O^*(2^{(1-\eta/2)v})$ bound.

For the randomized statement use randomized SETH to choose $k$, and amplify each index using Lemma~\ref{lem:rt-amplification}. Taking $K=O(v/(1/2-p)^2)$ makes each fixed query's error at most $2^{-4v}$. The complete collection of sparse formulas and partial-assignment queries is determined by the SAT input, independently of all answers. A union bound over this collection gives total error below $1/3$ for sufficiently large $v$; no independence between different queries is required. The factor $K$ is polynomial in $v$, so the same exponent contradiction holds. This proves the randomized claim under its own hypothesis.
\end{proof}

The order is essential: first fix the desired preprocessing exponent and query saving, and only then select the logarithmic dimension coefficient. This does not give hardness for every small fixed $c$.

In particular, the deterministic quantifiers are
\[
\begin{gathered}
 \forall b\ge1\ \forall\epsilon\in(0,1)\ \exists c_0>0\\
 \forall c\ge c_0\ \exists n_0\ \forall n\ge n_0,
\end{gathered}
\]
where $b,\epsilon,c$ are fixed before $n$ varies; the randomized statement fixes $p$ before choosing $c_0$. For a fixed small $c$, storing the answers to all $2^d=n^{c+o(1)}$ Boolean queries takes at most $O(nd2^d)=n^{c+1+o(1)}$ preprocessing time and permits direct lookup. This is why an arbitrary polynomial preprocessing budget cannot be excluded by fixing one universal dimension coefficient.

\section{Reduction framework}
\label{fw:framework}

The reduction has two maps: the stored vectors determine one long target, and an arriving vector determines one short query and its threshold. We first give these maps for the common numerical construction. We then prove the indexing transfer once and collect the witness and numerical lemmas used to instantiate it. Score-specific sections change the encoding only where their recurrence requires it.

\subsection{The common encoding}\label{fw:construction}

The common numerical construction starts with the following bounded embedding. Its two encoding families are asymmetric: only cross-family distances encode orthogonality.
\begin{lemma}[Bounded cross-family embedding]\label{lemma:raw-embedding}\label{lem:raw-embedding}
For $u,v\in\{0,1\}^d$, concatenate the coordinate gadgets
\[
\begin{array}{ll}
P_0=(0,0),&P_1=(2,0),\\
R_0=(1,0),&R_1=(0,1)
\end{array}
\]
to obtain $P(u),R(v)\in\{0,1,2\}^{2d}$ in $O(d)$ time. Then
\begin{align}
\|P(u)-R(v)\|_1&=d+2\langle u,v\rangle,\label{rs:embedding-l1}\\
\|P(u)-R(v)\|_\infty&=
\begin{cases}1,&\langle u,v\rangle=0,\\2,&\langle u,v\rangle>0.\end{cases}
\label{rs:embedding-linf}
\end{align}
For every common scale $s>0$, translation $b\in\mathbb R$, and $p\in\{1,\infty\}$,
\[
\|(sP(u)+b\mathbf1)-(sR(v)+b\mathbf1)\|_p
=s\|P(u)-R(v)\|_p.
\]
\end{lemma}
\begin{proof}
The four Boolean pairs give local $L_1$ cost $1+2u_iv_i$ and local $L_\infty$ cost $1+u_iv_i$. Sum the former and maximize the latter over coordinates. Translation cancels in every difference; positive scaling multiplies its absolute value by $s$.
\end{proof}

For $z\in[0,\Delta]^{2d}$ and $C>\Delta$, define
\[
B_C(z)=(C,z_1,2C,z_2,\ldots,2dC,z_{2d},(2d+1)C).
\]
Its length is $m=4d+1$. Given the ordered family $A=(u_1,\ldots,u_n)$, concatenate $S=B_C(P(u_1))\circ\cdots\circ B_C(P(u_n))$, and set $Q=B_C(R(v))$. A \emph{canonical} window is a whole block; every other length-$m$ window starts at a nonzero phase within a block. Paths start at $(1,1)$, end at $(m,m)$, and use steps $(1,0),(0,1),(1,1)$; a band restricts visited cells to $|i-j|\le W$.

For the unweighted absolute-cost construction, the two maps are therefore
\begin{align*}
 A&\longmapsto B_C(P(u_1))\circ\cdots\circ B_C(P(u_n)),\\
 q&\longmapsto B_C(R(q)).
\end{align*}
The target has $n$ blocks and is fixed before $q$ is known. Each canonical window is compared with the query through one stored vector. The common anchors encode coordinate positions; the bounded data values encode orthogonality. On the canonical diagonal the anchors cancel, leaving $d+2\langle u,q\rangle$ under absolute summation and $1$ versus $2$ under bottleneck aggregation. Choosing the anchor separation then makes an unintended local match exceed the decision threshold.

For example, at $d=1$ with $C=4$, the data blocks for bits zero and one are $(4,0,8,0,12)$ and $(4,2,8,0,12)$. The query for bit one is $(4,0,8,1,12)$. Their diagonal absolute costs are one and three. The first-departure argument in Section~\ref{sec:rt-raw} shows why a low-cost path cannot avoid this distinction, while its first-cell argument deals with every shifted window. The example illustrates the maps; the proof quantifies over all dimensions, stored vectors, and physical offsets.

For each score three tasks remain: exhibit an accepted canonical window in a Yes instance; exclude a nonorthogonal canonical witness; and account for all other windows. The third task is essential because a shifted window can combine parts of two vectors. The buffered constructions in Section~\ref{ds:section} use a slightly different form: any accepted window containing a complete core certifies an orthogonal core, while every window without one is rejected. Thus rejecting every noncanonical window even in a Yes instance is a sufficient template, not a necessary property of a reduction.

\subsection{Transfer to static indexing}\label{fw:transfer}

\begin{definition}[Uniform index-preserving reduction]\label{def:rt-reduction}
A reduction consists of a data encoder $A\mapsto S_A$ and a query encoder $q\mapsto(Q_q,\tau_q)$, with public dimension and model parameters. The first encoder does not use the query; the second does not inspect the stored data. The encoders are uniform for all sufficiently large $d$. Their construction times are $n^{1+o(1)}$ and $n^{o(1)}$, their lengths are at most $n^{1+o(1)}$ and $n^{o(1)}$, and these bounds are uniform over input contents. Every generated value and required construction operation is legal in the declared model. Finally, orthogonality is equivalent to the threshold predicate over \emph{all} physical windows, not only the designated block starts.
\end{definition}

\begin{lemma}[Index-preserving transfer]
\label{lem:rt-transfer}
Suppose a deterministic reduction constructs a target $S_A$ of length $N=n^{1+o(1)}$, in $n^{1+o(1)}$ time using only $A$, and constructs $(Q_q,\tau_q)$ in $n^{o(1)}$ time using only $q$ and public parameters. Suppose all encoded quantities occupy $O(1)$ words and \eqref{eq:rt-query} is equivalent to OnlineOV, including every window. For any fixed $\beta\ge1$ and $\delta\in(0,1)$, the respective SETH assumption excludes preprocessing $O(N^\beta)$ and query time $O(N^{1-\delta})$, provided the allowed query lengths include the reduction's hard family for the coefficient selected in Lemma~\ref{lem:rt-ovi}. The randomized version preserves every fixed joint per-input error bound $p<1/2$.
\end{lemma}
\begin{proof}
Fix $(\beta,\delta)$, and also $p$ when randomness is allowed. Apply Lemma~\ref{lem:rt-ovi} with
\[
 b=\beta+\tfrac12,\qquad \epsilon=\delta/2
\]
to select a sufficient fixed $c$, and put $d=\lceil c\log_2 n\rceil$. A hypothetical target index yields an OnlineOV preprocessing algorithm by first computing $S_A$ and then building the index on $S_A$. Its time is
\[
 n^{1+o(1)}+O(N^\beta)=n^{\beta+o(1)}.
\]
On the arrival of $q$, generate $(Q_q,\tau_q)$ and invoke the target query algorithm once. This takes
\[
 n^{o(1)}+O(N^{1-\delta})=n^{1-\delta+o(1)}.
\]
For sufficiently large $n$ these two bounds are respectively $O(n^{\beta+1/2})$ and $O(n^{1-\delta/2})$, contradicting the selected OnlineOV lower bound. The hard-length coverage premise ensures that every generated query is within the target index's correctness and time guarantees.

Both encoding maps are deterministic and independent of the target index's random coins. Their all-window equivalence therefore makes the composed OnlineOV answer incorrect exactly when the target index's answer is incorrect. This preserves the joint per-input error $p$ without a new amplification step. Worst-case running time over random choices is preserved as well. The reduction needs the returned Boolean answer, not an evaluation of the score by the encoder.
\end{proof}

More generally, if $N\le n^{1+o(1)}$ and $|Q|=n^{o(1)}$, then for any fixed $\eta_P>0$ and $0<\eta_Q<\delta$, sufficiently large $n$ satisfies $N^\beta=O(n^{\beta+\eta_P})$, $N^{1-\delta}=O(n^{1-\delta+\eta_Q})$, and $|Q|=O(n^{\eta_Q})$. Indeed, the vanishing exponent terms are eventually smaller than $\eta_P/\beta$, $\eta_Q/(1-\delta)$, and $\eta_Q$, respectively. Thus the exponent slack can be made arbitrarily small; the transfer proof uses one convenient fixed choice.

Here and below, ``hard length coverage'' means that there is $n_0$ such that $|Q_q|\in\mathcal L_{|S_A|}$ for every $n\ge n_0$ and every reduction input $(A,q)$, with the selected $c$ and all other external parameters fixed. Each construction below specifies its own $m=m_F(d)=d^{O(1)}$ and target size. For example, the normalized construction has $m=12d+1$, raw DTW has $m=4d+1$, and the inverse-derivative construction adds two physical endpoint samples. Correctness and time guarantees must hold uniformly over the declared length domain. The result applies to an index supporting all lengths, and to restricted domains with this coverage. It asserts neither that constant-length queries are hard nor an unconditional $\Omega(N)$ bound.

\subsection{Block encodings}\label{fw:blocks}

A data encoder concatenates query-independent containers $G(a^{(1)}),\ldots,G(a^{(n)})$, possibly with boundary strings.  Each container has a designated length-$m$ canonical window $X_i$; the query encoder produces $Q(q)$ of length $m$.  Containers, boundaries, query length, and their generation costs are $d^{O(1)}$ in our constructions, so $N=nd^{O(1)}=n^{1+o(1)}$ and query generation is $n^{o(1)}$. Canonical windows refer to designated starting positions.  Every other physical starting position must also be accounted for.

\begin{lemma}[Phase and cyclic-order obstruction]\label{fw:phase}
For a seamless concatenation of length-$m$ blocks, a length-$m$ window at offset $0<t<m$ is a suffix of length $m-t$ followed by the next block's prefix of length $t$. Suppose each block has distinct anchors $c_1,\ldots,c_r$, $r\ge2$, at common positions $1=p_1<\cdots<p_r=m$, with no data symbol equal to an anchor.  The query contains these anchors once in the same order and also has disjoint data symbols.  Every global order-preserving alignment of an offset window must delete an anchor, pair an anchor with data, or pair two different anchors.  For paths without gaps, one of the latter two events is unavoidable.
\end{lemma}
\begin{proof}
The two lengths follow by counting positions before and after the block boundary.  There is a unique $s$ with $p_s\le t<p_{s+1}$; the window's anchor order is $c_{s+1},\ldots,c_r,c_1,\ldots,c_s$. If none of the listed events occurs, all anchors on both sides must be paired with their same-named counterpart.  In particular, $c_1$ and $c_{s+1}$ would need to appear in the same order on both sides, a contradiction.  Repeated pairings cannot reverse their order.
\end{proof}

The certificate can be written as
\[
 \mathfrak C_F=(\mathsf{Enc}_A,\mathsf{Enc}_Q,
                   \mathsf{ori},m,\tau,\mathcal M),
\]
where $\mathsf{ori}$ records minimization or maximization and $\mathcal M$ specifies the computational model, including its parameter representation. For minimization its central obligation is
\begin{align*}
 &\exists a\in A:\langle a,q\rangle=0\\
 &\hspace{8mm}\Longleftrightarrow
 \min_{X\in\operatorname{Win}_{m}(S_A)}F(X,Q_q)\le\tau_q;
\end{align*}
for maximization replace $\min,\le$ with $\max,\ge$. Every offset is quantified, and each optimum ranges over all legal witnesses at that offset. The threshold is generated from the query and public parameters. The certificate also records the separate sizes and construction times of the two encoders, their uniformity in $d$, and the legality of all required numerical operations. These obligations are shared by the score-specific arguments below.

\begin{proposition}[Certificate closure]\label{fw:certificate}
A reduction certificate consists of its two encoders, orientation, query length, threshold, parameter domain, and machine model.  Besides the resource conditions of Lemma~\ref{lem:rt-transfer}, it must prove that an orthogonal vector gives an accepted canonical window, and that a No OnlineOV instance has \emph{no} accepted physical window. Such a certificate yields the corresponding deterministic or randomized indexing lower bound in its declared model and covered length domain.
\end{proposition}
\begin{proof}
Completeness gives the forward implication of the threshold predicate; the all-window No condition gives the reverse implication.  For a maximization score first apply Proposition~\ref{fw:elementary}, then Lemma~\ref{lem:rt-transfer}.
\end{proof}
A useful stronger sufficient template rejects each nonorthogonal canonical window and every noncanonical window separately.  It is not required: a shifted buffered window may retain a full data core and be accepted in a Yes instance.  A valid alternative proves that any accepted such window certifies an orthogonal core.  In either form, the proof covers \emph{every} physical offset and \emph{every} legal witness, including offsets within containers or boundary encodings. An order obstruction alone is not a numerical cost bound.

\subsection{Alignment bounds}

\begin{lemma}[Dominance, charging, projection, and initialization]
\label{fw:witness}
Let $\pi$ range over the legal global witnesses and let $\tau$ be the threshold.  Each of the following is a sufficient exclusion rule.
\begin{enumerate}
\item \emph{Dominance.} If a witness cost dominates each of its actual nonnegative local event costs, and every bad witness contains an event of cost at least $\Lambda>\tau$, no bad witness has cost at most $\tau$.
\item \emph{Cumulative charging.} If $\operatorname{cost}(\pi)\ge L+\sum_e\chi_e(\pi)$ with nonnegative charges, and every bad witness has total charge at least $\Gamma$ with $L+\Gamma>\tau$, every low-cost witness is good. For maximization replace this bound by $\operatorname{score}(\pi)\le U-\sum_e\chi_e(\pi)$ and require $U-\Gamma<\tau$.
\item \emph{Core projection.} Suppose every full witness projects to a legal core witness $P(\pi)$.  If $\operatorname{cost}_{\rm full}(\pi)\ge B+ \operatorname{cost}_{\rm core}(P(\pi))$, then $\operatorname{OPT}_{\rm full}\ge B+\operatorname{OPT}_{\rm core}$. For maximization, reversing both inequalities gives $\operatorname{OPT}_{\rm full}\le B+\operatorname{OPT}_{\rm core}$.
\item \emph{Initial barrier.} If every witness pays an unavoidable $B_0(X,Q)>\tau$ and all subsequent increments are nonnegative, the window is rejected.
\end{enumerate}
\end{lemma}
\begin{proof}
For dominance, the total is at least $\Lambda>\tau$. For charging it is at least $L+\Gamma$ (or at most $U-\Gamma$). For projection each projected witness is feasible for the core optimization; taking the appropriate extremum proves the stated bound. For initialization the remaining nonnegative sum cannot remove $B_0$.
\end{proof}
Sum and maximum aggregation satisfy dominance; averaging and soft-min need not.  Charges require a valid total-cost inequality, without counting the same operation more than the inequality allows.  A core projection must preserve consumption and order, not just discard unwanted costs.  Local restarts can bypass an initial barrier.  Thus cheap deletion, repeated matching, split/merge operations, and rewards must be checked against the actual recurrence for each score.

\subsection{Arithmetic and soft minima}\label{fw:arithmetic}

\begin{lemma}[Entropy-robust soft minimum]\label{fw:entropy}
Let a nonempty finite witness family satisfy $|\Pi(X,Q)|\le M(N)$, uniformly over all relevant windows.  For $\gamma>0$, put
\[
\begin{gathered}
F_\gamma=-\gamma\ln\sum_{\pi\in\Pi}e^{-C(\pi)/\gamma},\\
C^*=\min_{\pi\in\Pi}C(\pi).
\end{gathered}
\]
Then
\[
C^*-\gamma\ln M(N)\le F_\gamma\le C^*.
\]
Consequently, a Yes witness of cost at most $a$ and a uniform No lower bound $b$ give a strict soft-min threshold whenever $b-a>\gamma\ln M(N)$; choose $a\le\tau<b-\gamma\ln M(N)$. If every complete witness cost is scaled by $\kappa>0$, it suffices that $\kappa(b-a)>\gamma\ln M(N)$.
\end{lemma}
\begin{proof}
The partition sum lies between $e^{-C^*/\gamma}$ and $M(N)e^{-C^*/\gamma}$.  Apply the decreasing function $-\gamma\ln$. The Yes bound uses the upper inequality and the No bound the lower inequality.  Scaling all witness costs replaces $a,b$ by $\kappa a,\kappa b$.  Scaling data coordinates alone is insufficient for this last inference unless its effect on every witness is proved.
\end{proof}
Here $\ln$ denotes the natural logarithm.  The lemma is algebraic; it does not itself provide exact transcendental machine operations.

For an integer $z$ put $\operatorname{bits}(z)=1+\lceil\log_2(|z|+1)\rceil$. A rational is represented by an integer numerator and a positive integer denominator, with total bit length the sum of their lengths. Reduction legality concerns the required intermediate values as well as the final input: cross products before comparison and values before cancellation must also fit. Since $N=n^{1+o(1)}$ and $d=\lceil c\log_2 n\rceil$ for fixed $c$,
\[
\begin{gathered}
 \log N=\Theta(\log n),\qquad d=\Theta(\log N),\\
 \log d=\Theta(\log\log N).
\end{gathered}
\]

\begin{proposition}[Finite-word numerical criterion]\label{fw:bits}
For $N=n^{1+o(1)}$ and $d=\lceil c\log n\rceil$ with fixed $c$, assume every input, threshold, and intermediate quantity needed by the reduction is rational $p/q$, $q>0$, with $|p|,q\le2^{ad}d^b$ for fixed constants $a,b$. If each required primitive is implemented by $O(1)$ word operations, the reduction is valid on a $w=\Theta(\log N)$ Word-RAM. The values use $O(d)=O(\log N)$ bits; polynomial bounds in $d$ improve this to $O(\log d)=O(\log\log N)$ bits.
\end{proposition}
\begin{proof}
Taking binary logarithms gives $O(d+\log d)=O(d)$ bits for each numerator and denominator, and $d=\Theta(\log N)$.  Thus each value occupies a constant number of words.  The assumed primitive implementation supplies the operation bound.  Under polynomial amplitudes the same calculation gives $O(\log d)$ bits.
\end{proof}
Polynomially many integer additions, or rational additions with a common bounded denominator, increase bit length only by $O(\log d)$; a fixed number of multiplications preserves $O(d)$ bits.  General unequal-denominator sums require a separate bound on final and intermediate numerators and denominators, including cross-products before cancellation.  A constant number of words is not automatically one word, nor does representation alone establish a primitive's time.

\begin{corollary}[Bounded algebraic closure]\label{fw:algebraic-closure}
Suppose atomic numerators and denominators have magnitude $2^{O(d)}$. Consider quantities formed by $d^{O(1)}$ additions, subtractions, absolute values, maxima and minima, together with a fixed number of multiplications. If all final and required intermediate numerators and denominators still have magnitude $2^{O(d)}$, each is representable in $O(1)$ machine words on the hard family. The same assertion with polynomial-in-$d$ magnitude gives $O(\log d)$ bits.
\end{corollary}
\begin{proof}[Proof sketch]
Apply Proposition~\ref{fw:bits} to the asserted magnitude bounds. Integer or common-denominator sums of polynomially many terms add only $O(\log d)$ bits; fixed-degree products add only a constant factor to the bit bound. For unrelated rational denominators the magnitude premise must be checked explicitly, including intermediate cross products. This is a representation conclusion; the primitive-operation condition remains part of the model.
\end{proof}

\begin{proposition}[Uniform numerical error]\label{fw:precision}
Suppose Yes scores satisfy $F\le a$ and No scores satisfy $F\ge b$. If a computed score satisfies $|\widehat F-F|\le\varepsilon$ on \emph{all} reduction instances, and $b-a>2\varepsilon$, comparison with $(a+b)/2$ separates them, assuming these quantities can be represented and compared in the declared model. For a soft-min score with underlying cost bounds $a,b$, the sufficient condition is $b-a>\gamma\ln M+2\varepsilon$.
\end{proposition}
\begin{proof}
The computed Yes value is at most $a+\varepsilon$ and the No value is at least $b-\varepsilon$; the midpoint lies strictly between them. For soft-min substitute the bounds $a$ and $b-\gamma\ln M$ from Lemma~\ref{fw:entropy}, and use their midpoint $(a+b-\gamma\ln M)/2$.
\end{proof}
This is a sufficient certificate for a score-estimation implementation, not a necessary output of an exact Boolean index.  An exact predicate reduction generates finite inputs and thresholds and invokes the assumed index; it need not evaluate the score itself.  In particular, an $\exp/\ln$ formula does not prevent a Word-RAM lower bound for its mathematically exact threshold predicate.  Excluding a stronger machine with unit-cost exact transcendental primitives requires an OnlineOV assumption in that stronger model.  Numerical error and randomized answer error are separate quantities.

\subsection{Decision and approximation}

\begin{proposition}[Elementary transfer rules]\label{fw:elementary}
The conclusion of Lemma~\ref{lem:rt-transfer} also applies to the following variants, subject to its size, uniformity, model, and length-domain conditions.
\begin{enumerate}
\item A finite-valued maximization score $H$ with threshold $\tau$ can be replaced by $D_H=-H$ with threshold $-\tau$.
\item A hard slice with a parameter fixed independently of $n,d$ implies hardness of the uniform problem accepting that parameter as input. Hardness for every fixed parameter requires a reduction for each such parameter; its constants may depend on that parameter.
\item A Yes bound $D_{\rm sub}\le\tau$ and a No bound $D_{\rm sub}\ge\tau+\Delta(N)$ suffice whenever $\Delta(N)>0$ and the threshold predicate is exact, even if $\Delta(N)\to0$.
\item A threshold-sensitive witness-or-failure output, an exact optimum, or a complete report is at least as hard as decision, provided its claimed time includes the necessary output and decision costs.
\end{enumerate}
\end{proposition}
\begin{proof}
For the first claim, $\min_X(-H(X,Q))=-\max_XH(X,Q)$. For the second, restrict the uniform algorithm to the hard slice. For the third, a positive gap preserves the strict No inequality. For the last, test failure, compare the returned optimum with the threshold, or test whether the report is empty.  A returned location alone must either include the threshold guarantee or permit verification within the claimed time.  These transformations preserve the stated error probability.
\end{proof}

An exact lower bound alone gives no constant-factor approximation lower bound.  Rubinstein's independent ANN theorem states that, under SETH, for fixed preprocessing and query exponents $\beta>0$ and $\delta>0$, some $\epsilon>0$ rules out deterministic $(1+\epsilon)$-ANN with $O(n^\beta)$ preprocessing and $O(n^{1-\delta})$ queries under Euclidean, Manhattan, Hamming, and edit distance~\cite{rubinstein2018hardness}. Its transfer to another score needs an additional reduction.

\begin{proposition}[Multiplicative-gap interface]\label{fw:approx}
Suppose a reduction satisfies all the conditions of Lemma~\ref{lem:rt-transfer} except that its decision equivalence is replaced by the uniformly comparable bounds
\[
\text{Yes: }D_{\rm sub}\le\tau_q,\qquad
\text{No: }D_{\rm sub}>g(N)\tau_q,
\]
where $D_{\rm sub}\ge0$, $\tau_q>0$, and $g(N)>1$. Then the same conditional lower bound holds for this gap decision problem, whose answer is unrestricted between its two promise regions. It also holds for an algorithm returning a feasible cost $\widehat D$ with
\[
D_{\rm sub}\le\widehat D\le\rho(N)D_{\rm sub},
\qquad \rho(N)<g(N)
\]
for all sufficiently large $N$, when the returned cost can be compared with $g(N)\tau_q$ within the transfer's query budget.
\end{proposition}
\begin{proof}
Every reduction instance lies in a promise region.  In a Yes instance, $\widehat D\le\rho(N)\tau_q<g(N)\tau_q$; in a No instance, $\widehat D\ge D_{\rm sub}>g(N)\tau_q$.  This comparator answers OnlineOV, so the transfer lemma applies with unchanged error semantics.
\end{proof}
For example, a Yes threshold $\tau_N=\Theta(\log N)$ and a No lower bound $\tau_N+\kappa$, for constant $\kappa>0$, permit
\[
g(N)=1+\frac{(1-\eta)\kappa}{\tau_N}
     =1+\Theta(1/\log N),\qquad 0<\eta<1.
\]
Choose a fixed representable $\eta$; the No inequality is then strict. This excludes suitable high-accuracy ratios, not an arbitrary fixed constant ratio.

\begin{table*}[tp]
\centering
\small
\setlength{\tabcolsep}{4pt}
\renewcommand{\arraystretch}{1.12}
\begin{tabularx}{\textwidth}{@{}p{0.19\textwidth}p{0.12\textwidth}p{0.21\textwidth}p{0.22\textwidth}X@{}}
\toprule
Score & Query length & Yes value or bound & Bound in a No instance & Certificate \\
\midrule
Raw DTW / cDTW & $4d+1$ & $\min D\le d$ & $\min D\ge d+1$ & Theorem~\ref{thm:rt-raw} \\
Raw DFD & $4d+1$ & $\min D=1$ & $\min D=2$ & Theorem~\ref{thm:rt-raw} \\
cWDTW & $4d+1$ & $\min D\le\omega_0d$ & $\min D>\omega_0d$ & Corollary~\ref{rs:weighted-itakura} \\
ERP & $4d+1$ & $\min D\le d$ & $\min D>d$ & Lemma~\ref{em:erp-separation} \\
TWED & $4d+1$ & $\min D\le2d$ & $\min D>2d$ & Lemma~\ref{em:twed-separation} \\
MSM & $4d+1$ & $\min D\le d$ & $\min D>d$ & Lemma~\ref{em:msm-separation} \\
cDDTW (physical) & $4d+3$ & $\min D\le d$ & $\min D>d$ & Theorem~\ref{dc:cddtw} \\
LCSS & $2d^2+4d$ & $\max L=\tau_L$ & $\max L\le\tau_L-1$ & Lemma~\ref{ds:lcss-separation} \\
EDR & $7d$ & $\min E\le d$ & $\min E\ge d+1$ & Theorem~\ref{ds:edr-separation} \\
Swale & $7d$ & $\max W=4dr$ & $\max W\le(4d-3)r$ & Theorem~\ref{ds:swale-separation} \\
TS-SW & $4d+1$ & $\max T=4m-1$ & $\max T\le4m-1-2/d$ & Theorem~\ref{ds:tssw-separation} \\
Raw squared CID & $4d+1$ & $\min D^2<d+1$ & $\min D^2\ge d+4$ & Theorem~\ref{dc:cid} \\
Soft-DTW & $4d+1$ & $\min F_\gamma\le\kappa d$ & $\min F_\gamma>\kappa d+\kappa$ & Lemma~\ref{rs:soft-gap} \\
Normalized squared DTW & $12d+1$ & $\min D=4d/v$ & $\min D\ge(4d+4)/v$ & Theorem~\ref{thm:z:identity} \\
Normalized squared DFD & $12d+1$ & $\min D=1/v$ & $\min D=4/v$ & Theorem~\ref{thm:z:identity} \\
Normalized DFD & $12d+1$ & $\min D=1/\sqrt v$ & $\min D=2/\sqrt v$ & Corollary~\ref{cor:z:rational} \\
Rounded normalized squared scores & $12d+1$ & Below $T_D$ or $T_F$ & Above $T_D$ or $T_F$ & Corollary~\ref{cor:z:rounding} \\
\bottomrule
\end{tabularx}
\caption{Quantitative all-window separations. Minima and maxima range over every equal-length window. Here $\tau_L=2d^2+3d$, $r>0$ is the fixed Swale reward, $\omega_0$ is the diagonal weight, $\kappa$ is the integer scale in~\eqref{rs:soft-constants}, and $v$ is the common variance of canonical normalized blocks. The rounded thresholds are $T_D=(4d+2)/v$ and $T_F=5/(2v)$. An admissible Itakura restriction inherits the raw DTW row. Parameter and model restrictions are those in Table~\ref{in:catalog} and the cited statements.}
\label{tab:gap-catalog}
\end{table*}

\section{DTW and discrete Fr\'echet}
\label{sec:rt-raw}

We apply the alternating encoding of Section~\ref{fw:construction}. Its first-departure and first-cell bounds are the representative all-window argument for global path scores.

\begin{lemma}[First departure and first cell]\label{lem:rt-barrier}
For alternating blocks with common anchors and data in $[0,\Delta]$, every non-diagonal endpoint-to-endpoint path in a canonical window contains a local absolute difference at least $C-\Delta$. Every path in a noncanonical window contains such a difference at its first cell. These statements hold before imposing any band restriction.
\end{lemma}
\begin{proof}
Fix a canonical block and a nondiagonal path. Before its first horizontal or vertical step the path has visited only $(1,1),(2,2),\ldots,(r,r)$ for some $r<m$. That departure enters $(r+1,r)$ or $(r,r+1)$; a diagonal step cannot be the first departure. Odd positions contain anchors and even positions contain data. If $r$ is odd, the departure pairs the data at position $r+1$ with an anchor at position $r$. If $r$ is even, it pairs the next anchor with data at position $r$. In both directions and both parities the anchor has value at least $C$ and the data value is at most $\Delta$, proving the first bound. It is a statement about an actual visited cell, not merely about the order of two anchors.

Now let a window start at offset $t\in\{1,\ldots,m-1\}$ of a block. The query's first value is $C$. The window's first value is either data in $[0,\Delta]$ or a later anchor $jC$ with $j\ge2$. In the former case its difference from $C$ is at least $C-\Delta$, and in the latter case it is at least $C$. Every global path visits the first cell $(1,1)$, independently of the remaining contents of the window. These are all possible starting positions in the concatenated target, so the bound includes windows crossing a block boundary. Both arguments hold in the unrestricted path family and therefore in any allowed subset.
\end{proof}

Here DTW uses the sum of absolute local differences, unit state weights, fixed endpoints, and no division by path length~\cite{berndt1994using,sakoe1978dynamic}. DFD minimizes the maximum absolute difference on the same monotone grid paths~\cite{eiter1994computing}. Writing $\ell_{ij}=|x_i-q_j|$, their recurrences are
\begin{align*}
 D_{ij}&=\ell_{ij}+\min\{D_{i-1,j},D_{i,j-1},D_{i-1,j-1}\},\\
 F_{ij}&=\max\{\ell_{ij},\min\{F_{i-1,j},F_{i,j-1},F_{i-1,j-1}\}\}.
\end{align*}
Initialize the first cell with its local cost and use only legal predecessors; cells outside a chosen band are unavailable. The output is the state $(m,m)$. The path family is finite and nonempty, so both minima are attained. One window is evaluated in $O(m^2)$ time. Lemma~\ref{lem:rt-barrier} is common to the sum and bottleneck recurrences; only the threshold comparison changes.

\begin{theorem}[Raw DTW and DFD]\label{thm:rt-raw}
For every fixed integer band $W\ge0$, and also for the unconstrained path family, the exact equal-length-window DTW and DFD problems have the conditional indexing lower bounds of Lemma~\ref{lem:rt-transfer} under the respective deterministic and randomized SETH assumptions. Their hard lengths are $m=4\lceil c\log n\rceil+1$.
\end{theorem}
\begin{proof}[Proof for DTW]
For DTW set $\Delta=2$, $C=d+3$, and $\tau=d$. The encoders use the same alternating block shape; data blocks depend only on the stored vectors and the query block depends only on $v$. On a canonical window the full diagonal has zero anchor costs and data cost
\[
 \sum_{j=1}^{2d}|P(u)_j-R(v)_j|=d+2\langle u,v\rangle.
\]
Thus an orthogonal data vector gives a legal path of cost $d$. This proves completeness without assuming that every path is diagonal.

For soundness, fix any window and any legal path. In a canonical window a nondiagonal path contains a first-departure cell of cost at least $C-2=d+1$ by Lemma~\ref{lem:rt-barrier}; nonnegative summation already exceeds the threshold. A remaining canonical path is diagonal and, in a No instance, costs at least $d+2$. In a noncanonical window the mandatory first cell has cost at least $d+1$, regardless of the two data blocks that the window intersects. These cases exhaust all windows and all paths. We force only \emph{threshold-feasible} paths to be diagonal; an optimal path in a No block need not be diagonal.

Every Yes witness is diagonal and therefore satisfies each fixed nonnegative band. A path restriction cannot introduce an accepted No witness. The construction has $N=n(4d+1)$, query length $m=4d+1$, and construction times $O(nd)$ and $O(d)$. Its symbols have polynomial magnitude in $d$. Finite local costs and DP sums also have polynomial magnitude, because a path uses at most $2m-1$ cells. Lemma~\ref{lem:rt-transfer} supplies both indexing exclusions. If bandwidth is an explicit query parameter, restrict the proposed uniform index to one fixed band.
\end{proof}
\begin{proof}[Proof sketch for DFD]
Use the same encoders with $C=4$ and threshold $1$. The diagonal bottleneck is $1$ for an orthogonal pair and $2$ otherwise by~\eqref{rs:embedding-linf}. Every nondiagonal canonical path and every shifted-window path has a cell of cost at least $C-2=2$. A bottleneck maximum dominates that cell, just as the nonnegative path sum did in the DTW proof. Thus the all-window optimum is $1$ in Yes instances and $2$ in No instances, with the upper bound $2$ supplied by any canonical diagonal. The same path-restriction and resource arguments finish the transfer.
\end{proof}

\begin{corollary}[DFD approximation below two]\label{rs:dfd-approx}
Fix $1\le\rho<2$. The same lower bounds exclude an index returning a valid window $\widehat X$ with
\[
D_{\rm DFD}(\widehat X,Q)\le\rho\min_X D_{\rm DFD}(X,Q),
\]
or a valid window together with a feasible coupling whose bottleneck satisfies that bound. In the randomized case, feasibility and the approximation guarantee must hold jointly with probability at least $1-p$ for every fixed input, where $p<1/2$ is fixed.
\end{corollary}
\begin{proof}
Choose a fixed rational $g$ with $\rho<g<2$. Retain $S$ during preprocessing at an additional $O(N)$ cost. For a returned window, compute its exact DFD in $O(m^2)$ time and test whether it is below $g$. For an explicitly returned coupling, one may instead verify its feasibility and compute its bottleneck in $O(m)$ time. On success, the Yes value is at most $\rho<g$, whereas every No feasible value is at least $2>g$. Invalid output is treated as failure. The extra $\operatorname{poly}(d)=n^{o(1)}$ query work is absorbed by the exponent slack in Lemma~\ref{lem:rt-transfer}, and the error is at most $p$.
\end{proof}
This statement does not cover arbitrary underestimating values or expected approximation guarantees. At $\rho=2$ the present gap no longer separates the cases: on the constructed instances, any canonical diagonal is already a feasible $2$-approximation. This observation is only about the construction, not a general $2$-approximate index.

\begin{corollary}[Positive weights and fixed path restrictions]
\label{rs:weighted-itakura}
The same conditional lower bounds hold for cWDTW with local cost $\omega(|i-j|)|x_i-q_j|$ if the rule $\omega$ is fixed independently of input size, $\omega(0)=\omega_0>0$, and $\inf_r\omega(r)\ge\omega_{\min}>0$ for a fixed constant $\omega_{\min}$. Its required weights must be uniformly exactly generated and represented in $O(1)$ words, and required multiplication and comparison must take $O(1)$ Word-RAM operations. They also hold for any fixed Itakura-style path rule whose permitted paths form a subset of the standard monotone paths and contain the complete diagonal at every sufficiently large hard length.
\end{corollary}
\begin{proof}[Proof sketch]
Reuse Lemma~\ref{lem:rt-barrier} and the DTW proof, choosing a fixed positive integer $K$ with $K\omega_{\min}>\omega_0$, scale $C=2+Kd+1$, and threshold $\omega_0d$. The diagonal costs $\omega_0(d+2\langle u,v\rangle)$, while an excluded path incurs a weighted cell cost at least $\omega_{\min}(Kd+1)>\omega_0d$. The weight-generation premise supplies exact constant-word arithmetic. For the fixed path restriction, the canonical diagonal remains a Yes witness and its path set is a subset of the unrestricted No case. No new all-window or index-transfer argument is needed.
\end{proof}
The length-dependent modified logistic weights of the original WDTW proposal \cite{jeong2011weighted} need not satisfy the uniform positive lower bound, and their transcendental evaluation does not automatically meet the exact Word-RAM condition. Likewise, the historical Itakura formulation \cite{itakura1975minimum}, modern window constraints \cite{keogh2002exact}, and alternative step weights must not be identified without checking the stated path-subset premise. None of these results silently replaces absolute loss by squared loss or a path-length average.

\section{Edit distances}
\label{em:section}

The edit operations of these three measures require different locking arguments.  ERP charges the magnitude of a deleted anchor; TWED charges same-side variation; MSM can charge only its fixed penalty when an element lies between its two neighbors.  We therefore verify their actual recurrences separately.

Throughout this section, let $p=P(u)$ and $q=R(v)$ be the concatenated $2d$-coordinate codes of Lemma~\ref{lem:raw-embedding}.  Thus all data coordinates lie in $[0,\Delta]$, where $\Delta=2$, and
\begin{equation}
 \|p-q\|_1=d+2\langle u,v\rangle.
 \label{em:embedding}
\end{equation}
For a scale $C>\Delta$, set $c_j=jC$ and reuse the alternating block from Section~\ref{fw:construction}:
\begin{equation}
 \Phi_C(p):=B_C(p),\qquad m=4d+1.
 \label{em:block}
\end{equation}
Its length is $m=4d+1$.  Given $U=(u^{(1)},\ldots,u^{(n)})$, concatenate these blocks without separators to obtain $S$, and let $Q=\Phi_C(R(v))$. An aligned block is a \emph{canonical} window.  Every other length-$m$ window consists of a proper block suffix and the next block's prefix. In particular, its first element is either a data value or $c_j$ for $j\ge2$, and hence differs from $Q_1=C$ by at least $C-\Delta$.

\subsection{ERP}
\label{em:erp}

For gap value $g$, one-dimensional ERP~\cite{chen2004marriage} is the minimum cost of a global order-preserving edit alignment.  With prefix states $E(i,j)$, its recurrence is
\begin{equation}
 E(i,j)=\min\begin{cases}
 E(i-1,j-1)+|x_i-q_j|,\\
 E(i-1,j)+|x_i-g|,\\
 E(i,j-1)+|q_j-g|,
 \end{cases}
 \label{em:erp-dp}
\end{equation}
with $E(0,0)=0$, $E(i,0)=\sum_{a=1}^i|x_a-g|$, and $E(0,j)=\sum_{b=1}^j|q_b-g|$.  The output is $E(m,m)$. The reduction fixes $g=0$, $C=d+\Delta+1=d+3$, and threshold $d$. An alignment consists of columns $(x_i,q_j)$, $(x_i,g)$, or $(g,q_j)$; each input position is consumed exactly once.  A direct anchor match means the actual column $(c_j,c_j)$, not merely visiting a DP state.

\begin{lemma}[ERP locking and all-window separation]
\label{em:erp-separation}
Every canonical alignment of cost at most $d$ directly matches all same-name anchors.  Moreover,
\[
 \begin{aligned}
 &\exists u\in U:\langle u,v\rangle=0\\
 &\quad\Longleftrightarrow\quad
 \min_{X\in\operatorname{Win}_m(S)}\operatorname{ERP}_0(X,Q)\le d.
 \end{aligned}
\]
\end{lemma}
\begin{proof}
First consider a canonical alignment. Every query position occurs in exactly one alignment column. For a query anchor $c_j$, that column has one of three forbidden forms unless it is the direct match $(c_j,c_j)$: a gap column costs $c_j\ge C>d$; an anchor--data column costs at least $C-\Delta=d+1$; an unequal-anchor column costs at least $C>d$. Nonnegativity of all columns implies that a total-cost-$d$ alignment contains every direct same-name anchor column. This conclusion concerns the actual consumed pair, not just the presence of a DP state with matching indices.

These anchor columns occur in increasing order. Between $(c_k,c_k)$ and $(c_{k+1},c_{k+1})$, each side contains exactly one data element, $p_k$ and $q_k$. The legal subalignment either substitutes one for the other or consumes both through separate gap columns. The triangle inequality gives
\[
 |p_k-q_k|\le |p_k-g|+|q_k-g|.
\]
Thus the minimum chamber cost is exactly $|p_k-q_k|$. Chambers use disjoint alignment columns, so their lower bounds sum without counting any operation twice. A canonical threshold-feasible alignment consequently has cost at least $\|p-q\|_1=d+2\langle u,v\rangle$. Conversely the diagonal supplies exactly this cost. In particular an orthogonal canonical block passes, and a nonorthogonal canonical block cannot pass.

It remains to treat every shifted window. Put $r=2d+1$ and let $0<t<m$ be its offset from the preceding block start. The window is the suffix of one block followed by the corresponding-length prefix of the next. Since the two blocks put the same anchors at the same positions, its anchor order is
\[
 c_{s+1},\ldots,c_r,c_1,\ldots,c_s
 \qquad\text{for some }1\le s<r.
\]
For clarity, this covers a cut at either an anchor or a data position: $s$ counts the anchors in the removed prefix. Suppose there were a shifted alignment of cost at most $d$. The same three column-cost inequalities would force every query anchor to match its unique same-name window anchor. But the pair $c_{s+1},c_1$ appears in opposite orders in the window and query. A global order-preserving alignment cannot contain both direct matches. Hence some anchor must be consumed by one of the forbidden columns, already costing more than $d$. This contradiction is independent of the data values in the two neighboring blocks. Combining the canonical and shifted cases proves both directions of the all-window equivalence.
\end{proof}

The gap-$0$ restriction is sufficient for the main lower bound. For any fixed exactly representable $g$, translating \emph{every} data value and anchor by $g$, and using gap $g$, leaves every term of \eqref{em:erp-dp} unchanged.  This gives the corresponding fixed-gap slice whenever the translation is legal in the computational model. Translating the data alone while leaving the gap fixed would not justify this extension.

\subsection{TWED}
\label{em:twed}

We use the following unit-timestamp $L_1$ slice of TWED~\cite{marteau2009time}.  For each query and each candidate window, reset timestamps to $t_i=i$ and set $x_0=q_0=0$.  Fix $\nu,\lambda\ge0$ independently of the input size.  Define
\begin{align*}
 A_{ij}&=|x_i-q_j|+|x_{i-1}-q_{j-1}|+2\nu|i-j|,\\
 B_i&=|x_i-x_{i-1}|+\nu+\lambda,\\
 F_j&=|q_j-q_{j-1}|+\nu+\lambda.
\end{align*}
Then
\begin{equation}
 T(i,j)=\min\begin{cases}
 T(i-1,j-1)+A_{ij},\\
 T(i-1,j)+B_i,\\
 T(i,j-1)+F_j.
 \end{cases}
 \label{em:twed-dp}
\end{equation}
Initialize $T(0,0)=0$ and $T(i,0)=T(0,j)=+\infty$ for positive $i,j$.  A finite witness is a predecessor chain to $(m,m)$, so its first step is necessarily $(0,0)\to(1,1)$. These boundary conditions and local timestamp resetting are part of the problem definition; we make no claim for free endpoints, finite boundary initialization, or arbitrary timestamp conventions.  Zero parameters are permitted as scoring slices without asserting metric identity on the full space of time-stamped sequences.

Use \eqref{em:block} with $C=2d+\Delta+1=2d+3$ and threshold $2d$.

\begin{lemma}[TWED all-window separation]
\label{em:twed-separation}
For the stated TWED slice,
\[
 \exists u\in U:\langle u,v\rangle=0
 \quad\Longleftrightarrow\quad
 \min_{X\in\operatorname{Win}_m(S)}T(X,Q)\le2d.
\]
\end{lemma}
\begin{proof}[Proof sketch]
We use the initial-barrier and nonnegative-cost rules of Lemma~\ref{fw:witness}, checking the TWED increments directly. An adjacent anchor--data pair has variation at least $C-\Delta=2d+1$; the boundary pair $c_{2d+1},c_1$ has variation $2dC$. Thus every horizontal or vertical transition after the first step costs more than $2d$. The infinite boundary states force that first step to be $(0,0)\to(1,1)$, so a threshold-feasible canonical chain must be completely diagonal.

On this diagonal the time penalties vanish. Every data discrepancy appears once in the current term and once in the predecessor term; the virtual initial and final anchor discrepancies are zero. The cost is $2\|p-q\|_1=2d+4\langle u,v\rangle$, which gives the canonical Yes/No distinction. For every shifted window, the forced first step alone costs at least $|X_1-Q_1|\ge C-\Delta>2d$. Nonnegative later increments cannot remove this barrier. These inequalities replace ERP's alignment-column argument and cover all window phases.
\end{proof}

\subsection{MSM}
\label{em:msm}

For fixed $\lambda>0$, the one-dimensional MSM recurrence~\cite{stefan2012move} uses the closed interval $I(b,t)=[\min(b,t),\max(b,t)]$ and the cost
\begin{equation}
 \mathcal C_\lambda(a,b,t)=\lambda+
 \begin{cases}
 0,&a\in I(b,t),\\
 \min\{|a-b|,|a-t|\},&\text{otherwise}.
 \end{cases}
 \label{em:msm-cost}
\end{equation}
Here $M(i,j)$ compares query prefix $q_1,\ldots,q_i$ to window prefix $x_1,\ldots,x_j$.  Initialize
\begin{align*}
 M(1,1)&=|q_1-x_1|,\\
 M(i,1)&=M(i-1,1)+\mathcal C_\lambda(q_i,q_{i-1},x_1),\\
 M(1,j)&=M(1,j-1)+\mathcal C_\lambda(x_j,q_1,x_{j-1}),
\end{align*}
where the latter two equations apply for indices at least $2$. For $i,j\ge2$,
\begin{equation}
 M(i,j)=\min\begin{cases}
 M(i-1,j-1)+|q_i-x_j|,\\
 M(i-1,j)+\mathcal C_\lambda(q_i,q_{i-1},x_j),\\
 M(i,j-1)+\mathcal C_\lambda(x_j,q_i,x_{j-1}).
 \end{cases}
 \label{em:msm-dp}
\end{equation}
The output is $M(m,m)$.  The proof concerns predecessor chains of this standard DP, with diagonal Move and vertical/horizontal Split/Merge steps.  It does not assume a normal form for arbitrary edit-operation scripts.  Although the formulas and argument also make sense for $\lambda=0$, we reserve ``MSM metric'' for $\lambda>0$.

Set $C=d+\Delta+\lambda+3$ in \eqref{em:block}, and use threshold $d$.

\begin{lemma}[MSM all-window separation]
\label{em:msm-separation}
Every canonical predecessor chain of cost at most $d$ is the diagonal Move chain.  Furthermore,
\[
 \exists u\in U:\langle u,v\rangle=0
 \quad\Longleftrightarrow\quad
 \min_{X\in\operatorname{Win}_m(S)}M(X,Q)\le d.
\]
\end{lemma}
\begin{proof}[Proof sketch]
Use the first-departure organization of Lemma~\ref{lem:rt-barrier}, now with the actual Split/Merge cost~\eqref{em:msm-cost}. At the first departure from $(r,r)$, if $r$ is odd then the two old elements are the same anchor and the new element is data below both. If $r$ is even, both old elements are data and the new element is an anchor above both. In either parity, for either a horizontal or vertical departure, the new element lies outside the interval of the two old elements. The nonconstant branch of $\mathcal C_\lambda$ therefore charges at least $\lambda+C-\Delta>d$. This includes departures on the initialized first row or column. All other increments are nonnegative.

Consequently every threshold-feasible canonical chain is diagonal, with cost $\|p-q\|_1=d+2\langle u,v\rangle$. A shifted window instead has mandatory initialization $|Q_1-X_1|\ge C-\Delta>d$, which no subsequent operation can cancel. These two cases prove the all-window equivalence. The argument does not require an optimal No-instance chain to be diagonal.
\end{proof}

\subsection{Indexing lower bounds}
\label{em:index}

\begin{theorem}[Exact indexing for the three editing measures]
\label{em:main}
Fix ERP gap $g=0$, the TWED slice above with $\nu,\lambda\ge0$, or MSM with $\lambda>0$.  Parameters are independent of input size and exactly implementable in the declared Word-RAM model; fixed rational parameters suffice.  Under deterministic SETH, for every fixed $\beta\ge1$ and $\delta\in(0,1)$, there is no uniform exact static index with preprocessing $O(N^\beta)$ and worst-case query time $O(N^{1-\delta})$ on a domain covering the corresponding hard lengths. Under randomized SETH, the same statement holds for randomized indexes with any fixed per-input error probability $p<1/2$.
\end{theorem}
\begin{proof}[Proof sketch]
Apply Lemma~\ref{lem:rt-transfer} once to the three preceding all-window certificates. Their data and query maps are independent, have construction times $O(nd)$ and $O(d)$, and give $N=n(4d+1)$ and $m=4d+1$. The scales are $O(d)$ and the largest anchors $O(d^2)$, with constants depending only on fixed parameters. For fixed rational parameters, a common constant denominator makes all input values, thresholds and finite DP values polynomial-in-$d$ integers. The TWED symbol $+\infty$ is a DP sentinel. More general fixed representations must satisfy the same exact-operation premise. These facts establish the common resource conditions, so the deterministic and randomized conclusions follow with the dimension coefficient selected after the target exponents and error bound.
\end{proof}

\section{Derivative DTW and CID}
\label{sec:dc:derivative}

\subsection{Endpoint-trimmed derivative DTW}\label{dc:definition}
For a physical sequence $X=(x_1,\ldots,x_M)$, use the internal Keogh derivative estimator
\[
\mathcal D(X)_i=\frac{x_{i+2}+2x_{i+1}-3x_i}{4},
\qquad 1\le i\le M-2.
\]
Our cDDTW is the endpoint-to-endpoint, absolute-local-cost DTW of these two derivative sequences, with fixed integer band $W\ge0$ or with no band. Endpoints are trimmed. This is a precisely specified variant of the derivative approach of Keogh and Pazzani \cite{keogh2001derivative}; it does not identify the present score with the squared-loss, endpoint-replication convention of their original implementation. Differentiation is performed independently on each candidate physical window.

Use the raw block $Z=B_C(P(u))$ from Lemma~\ref{lem:raw-embedding}, where $\Delta=2$, $C=d+3$, and $m=4d+1$. A physical realization $I(Z)$ has length $M=m+2$ and is constructed as follows. Starting with $x_1=x_2=0$, set
\begin{equation}\label{dc:inverse}
x_{k+2}=-2x_{k+1}+3x_k+4z_k\quad(1\le k\le m).
\end{equation}
Thus $\mathcal D(I(Z))=Z$ exactly. Construct $S=I(B_C(P(u_1)))\circ\cdots\circ I(B_C(P(u_n)))$ and $Q=I(B_C(R(v)))$; the inverse recurrence restarts with two zeros in every block. No extra separators are inserted.

\begin{lemma}[Inverse-recurrence bit bound]\label{dc:bits}
The physical coordinates and the intermediates in~\eqref{dc:inverse} are integers of $O(d)$ bits and are generated in $O(d)$ Word-RAM time per block when $d=\Theta(\log n)$.
\end{lemma}
\begin{proof}
Let $B=(2d+1)C$ bound every desired derivative and put $Y_k=\max_{r\le k+1}|x_r|$. Then $Y_1=0$ and $Y_{k+1}\le5Y_k+4B$, whence $Y_k\le B(5^{k-1}-1)$. All coordinates and recurrence intermediates have magnitude at most $O(B5^M)=2^{O(d)}$. A constant number of $\Theta(\log N)$-bit words suffices per coordinate; the constant may depend on the already fixed dimension coefficient. The recurrence takes $O(d)$ word operations.
\end{proof}

\begin{lemma}[The two boundary derivatives]\label{dc:boundary}
At a boundary from a completed physical block to the next two zeros, the only derivative values differing from internal prescribed values are
\[
b_1=(2x_M-3x_{M-1})/4,\qquad b_2=-3x_M/4.
\]
They satisfy $\min\{|b_1|,|b_2|\}=\Omega(C3^{4d})$; in particular both have magnitude greater than $B+d$ for all sufficiently large $d$.
\end{lemma}
\begin{proof}
Write $w_k=x_{k+1}-x_k$. Equation~\eqref{dc:inverse} gives $w_{k+1}=-3w_k+4z_k$ and $w_1=0$. Since $m$ is odd, all anchor contributions to $w=w_{m+1}$ are positive, and all data contributions are nonpositive. Keeping the first anchor and bounding every negative data term gives
\begin{align*}
w&=4\sum_{j=1}^{m}(-3)^{m-j}z_j\\
&>\left(4C-\frac{3\Delta}{2}\right)3^{m-1}
>\frac52 C3^{m-1}.
\end{align*}
Indeed, the negative terms have magnitude at most $4\Delta\sum_{r\ge1}3^{m-2r}=(3\Delta/2)3^{m-1}$. Let $Z_\Sigma=\sum_{j=1}^m z_j$, a nonnegative polynomial in $d$. Summing $4z_k=w_{k+1}+3w_k$ gives $x_M=Z_\Sigma+3w/4$ and $x_{M-1}=x_M-w$. Consequently
\[
b_1=9w/16-Z_\Sigma/4,\qquad
|b_2|=9w/16+3Z_\Sigma/4.
\]
For sufficiently large $d$, $Z_\Sigma\le w/4$, so both magnitudes are at least $w/2>B+d$. Discarding finitely many small dimensions does not affect the indexing reduction.
\end{proof}

\begin{theorem}[All physical windows for cDDTW]\label{dc:cddtw}
For the stated derivative variant, some length-$M$ window of $S$ has score at most $d$ against $Q$ if and only if some $u\in A$ is orthogonal to $v$. Hence the deterministic and randomized conditional indexing bounds of Lemma~\ref{lem:rt-transfer} hold, with hard lengths $M=4d+3$.
\end{theorem}
\begin{proof}
For a canonical physical window, differentiation recovers the raw block. Orthogonality gives a diagonal cost $d$. In a No instance its diagonal cost is $d+2\langle u,v\rangle\ge d+2$; any other path has cost at least $C-2=d+1$ by Lemma~\ref{lem:rt-barrier}.

Every noncanonical physical window crosses one block boundary. If the boundary is between positions $1$ and $2$ of this window, trimming retains $b_2$ as an internal derivative; if it is between $M-1$ and $M$, trimming retains $b_1$. Every other boundary position retains both. Thus an extreme derivative survives every shift, including those nearest the two endpoints. Every endpoint-to-endpoint path consumes this derivative, whereas every query derivative has absolute value at most $B$. Lemma~\ref{dc:boundary} forces a local cost greater than $d$, regardless of the bandwidth.

The data length is $N=nM$, and construction times are $O(nd)$ and $O(d)$. Derivatives of arbitrary integer physical windows lie in $\tfrac14\mathbb Z$; they need not be integers. Multiply local costs, sums, and the threshold by four to perform exact integer DP. Lemma~\ref{dc:bits} bounds these quantities by $2^{O(d)}$. Both resource and all-window requirements of the transfer lemma hold.
\end{proof}

\subsection{Raw CID}
\label{sec:dc:cid}
For positive-complexity sequences, define
\begin{align*}
\operatorname{CE}(X)^2&=\sum_{i=1}^{m-1}(x_{i+1}-x_i)^2,\\
\operatorname{CF}(X,Q)&=\frac{\max\{\operatorname{CE}(X),\operatorname{CE}(Q)\}}
{\min\{\operatorname{CE}(X),\operatorname{CE}(Q)\}},\\
D_{\rm CID}(X,Q)&=\|X-Q\|_2\operatorname{CF}(X,Q).
\end{align*}
We analyze this raw formula \cite{batista2011cid}, without independent window-wise Z-normalization. Its extension to zero-complexity inputs requires a convention not used here. Comparing $D_{\rm CID}^2$ to a rational $T\ge0$ is equivalent to comparing $D_{\rm CID}$ to $\sqrt T$. The squared comparison uses integer sums and cross multiplication, so it needs no square-root primitive.

Use a different Boolean embedding:
\[
U_0=(1,1),\quad U_1=(2,2),\qquad
V_0=(2,1),\quad V_1=(0,1).
\]
The four squared distances are $1,1,1,5$, respectively. Concatenating these pairs gives $\|U(u)-V(v)\|_2^2=d+4\langle u,v\rangle$.

This identity uses only cross-family pairs. For example the conflicting Boolean pair has squared cost $(2-0)^2+(2-1)^2=5$, whereas each other pair has squared cost one. No same-family isometry is needed. All data coordinates lie in $[0,\Delta]$ for $\Delta=2$. Form raw anchored blocks $B_C(U(u))$ and $B_C(V(v))$ with $C=d^2+10\Delta+10=d^2+30$. Their common length is $4d+1$.

\begin{lemma}[Complexity-factor control]\label{dc:cf}
For every pair of canonical blocks, $\operatorname{CF}^2<1+1/d$.
\end{lemma}
\begin{proof}
With $c_k=kC$, expansion of adjacent squared differences gives
\begin{align*}
\operatorname{CE}(B_C(z))^2&=S_C-\delta(z),\\
S_C&=C^2\sum_{k=1}^{2d}(2k^2+2k+1)
>\frac{16}{3}d^3C^2,\\
\delta(z)&=\sum_{k=1}^{2d}\bigl(2z_k(c_k+c_{k+1})-2z_k^2\bigr).
\end{align*}
Each summand in $\delta$ is nonnegative and increasing on $[0,\Delta]$, since $c_k+c_{k+1}>2\Delta$. Therefore
\[
0\le\delta(z)<8d(d+1)C\Delta=:\delta_{\max}.
\]
Moreover,
\begin{align*}
(d+1)\delta_{\max}
&=8d(d+1)^2C\Delta\\
&\le32d^3C\Delta
<\frac{16}{3}d^3C^2<S_C,
\end{align*}
where $C>6\Delta$. In particular $S_C-\delta_{\max}>0$, and the ratio of any two canonical complexity squares is bounded by $S_C/(S_C-\delta_{\max})<1+1/d$.
\end{proof}

\begin{theorem}[Raw CID separation]\label{dc:cid}
If an orthogonal pair exists, a canonical window satisfies $D_{\rm CID}^2<d+1$. In a No instance every window satisfies $D_{\rm CID}^2\ge d+4$. Thus exact squared CID decision has the conditional indexing bounds of Lemma~\ref{lem:rt-transfer}.
\end{theorem}
\begin{proof}
Every canonical block has positive complexity because adjacent anchor and data values differ. Every shifted window crosses a boundary and contains its two adjacent distinct anchors $(2d+1)C,C$; it too has positive complexity. Thus all comparisons lie within the domain.

An orthogonal canonical pair has squared Euclidean distance $d$; Lemma~\ref{dc:cf} makes its squared CID strictly less than $d+1$. A nonorthogonal canonical pair has squared Euclidean distance at least $d+4$, and $\operatorname{CF}\ge1$. In a shifted window the first sample is either data in $[0,\Delta]$ or a later anchor $jC$, $j\ge2$, whereas $Q_1=C$. Its squared Euclidean distance is at least $(C-\Delta)^2>d+4$, again a lower bound on squared CID.

Take the rational squared threshold $T=d+1$. All coordinates, adjacent-difference sums, Euclidean sums, and cross products are polynomial in $d$, hence have $O(\log d)$ bits. Data and query construction take $O(nd)$ and $O(d)$ time, and $N=n(4d+1)$. The transfer lemma applies to these exact finite-word comparisons.
\end{proof}

\begin{corollary}[Certified CID score error]\label{dc:cid-rounding}
Suppose the squared score of every compared window is estimated with absolute error at most a uniform $\varepsilon<3/2$. Comparing the minimum estimated score to $d+5/2$ still separates the constructed Yes and No instances.
\end{corollary}
\begin{proof}
Every Yes witness has estimated square below $d+1+\varepsilon<d+5/2$. Every No window has estimated square at least $d+4-\varepsilon>d+5/2$.
\end{proof}
This is a sufficient error certificate, not a guarantee for arbitrary floating-point evaluation. The exact index theorem does not rely on such an evaluation strategy.

For completeness, its exact predicate can be written using only integer arithmetic. With $E=\|X-Q\|_2^2$, $A=\operatorname{CE}(X)^2>0$, and $B=\operatorname{CE}(Q)^2>0$, the squared score is $E\max(A,B)/\min(A,B)$. For a rational threshold $T=t/u$ with $u>0$, compare
\[
 uE\max(A,B)\ \le\ t\min(A,B).
\]
All factors and products are polynomial in $d$ on the construction, including on shifted windows. This supplies the finite-word predicate without computing either complexity square root.

\section{Discrete alignment scores}
\label{ds:section}

This section treats LCSS, EDR, Swale, and a specified Smith--Waterman-type time-series score. A large numerical anchor cannot by itself penalize a unit-cost mismatch or a freely skipped symbol. We instead use distinct phase alphabets, repeated anchors, alignment projection, and the reward lost by a truncated local alignment. All symbols below are encoded as distinct integers. For any fixed, exactly representable tolerance $\epsilon\in[0,1)$, the relation $|x-y|\le\epsilon$ therefore means equality of symbols. The convention at $\epsilon=0$ is explicitly non-strict. None of the three symbolic constructions applies window-wise Z-normalization.

\subsection{LCSS}
\label{ds:lcss}

The time-series LCSS formulation of \citet{vlachos2002discovering} also includes a temporal constraint. Our slice maximizes the length of a common subsequence with the above tolerance and \emph{no additional time band}. For positive tolerance, strict and non-strict amplitude tests agree on this integer encoding; the zero-tolerance convention is as stated above. We write this score as $L(X,Q)$.

For each $k\in[d]$, take a private alphabet $\{A_k,B_k,C_k,D_k,E_k\}$ and set
\begin{align*}
 P_k(0)&=(A_k,B_k),&P_k(1)&=(C_k,D_k),\\
 Q_k(0)&=(A_k,C_k),&Q_k(1)&=(B_k,E_k).
\end{align*}
The four pairs share exactly one symbol except the pair $(1,1)$, which shares none. Thus their local LCSS is $1-u_kv_k$. Take fresh, mutually distinct anchors $c_1,\ldots,c_{d+1}$ and $M=2d$. Define
\begin{align*}
 X(u)&=c_1^M P_1(u_1)c_2^M\cdots P_d(u_d)c_{d+1}^M,\\
 Q(v)&=c_1^M Q_1(v_1)c_2^M\cdots Q_d(v_d)c_{d+1}^M,\\
 m&=(d+1)M+2d,\qquad \tau_L=(d+1)M+d.
\end{align*}
The target is the concatenation of the $n$ data blocks $X(u)$.

\begin{lemma}[All-window LCSS separation]
\label{ds:lcss-separation}
For a canonical block,
\[
 L(X(u),Q(v))=\tau_L-\langle u,v\rangle.
\]
Every noncanonical length-$m$ window $Y$ satisfies $L(Y,Q(v))\le\tau_L-1$, independently of the data vectors. Consequently,
\begin{align*}
 &\exists u\in U:\ \langle u,v\rangle=0\\
 &\hspace{12mm}\Longleftrightarrow
 \max_{Y\in\operatorname{Win}_m(S)}L(Y,Q(v))\ge\tau_L.
\end{align*}
\end{lemma}
\begin{proof}
Give each anchor run and each coordinate gadget its own phase, ordered as in the query. Distinct phases use disjoint alphabets. Within a canonical block, matches therefore cannot move from one coordinate to another or from data to an anchor. Each anchor run contributes at most $M$ and each coordinate contributes exactly $1-u_kv_k$ at its local optimum. Concatenating these local optima respects the common phase order, so the upper bound is attained and gives $\tau_L-\langle u,v\rangle$.

Write a shifted window as $A B$, a nonempty proper suffix of one block followed by a nonempty proper prefix of the next. The cut may split an anchor run or a coordinate gadget. Phases in $A$ are later than phases in $B$, with the sole possible exception of that split phase appearing at both ends. A matched subsequence has nondecreasing query phase ranks. Consequently, if it uses symbols from both $A$ and $B$, every chosen symbol must belong to the split phase: any strictly later phase chosen in $A$ would prevent returning to an earlier phase in $B$, and the reverse alternative is excluded by the cut positions.

There are now three exhaustive cases. If matches use only $A$, the removed nonempty prefix contains at least one copy of the initial anchor $c_1$. All other phases contribute at most their canonical bounds, and $c_1$ cannot be supplied by another phase, so the score is at most $\tau_L-1$. If matches use only $B$, the missing suffix contains at least one copy of the final anchor and gives the same bound. If both parts are used, the query has at most $M$ occurrences of the split anchor, or at most two symbols of the split coordinate phase. Its score is therefore at most $\max\{M,2\}<\tau_L$. These bounds also cover a cut exactly between phases, where using both parts is impossible. Scores are integers, hence every shifted window has score at most $\tau_L-1$.

An orthogonal stored vector gives a canonical window of score $\tau_L$. In a No instance canonical windows lose at least one coordinate match and shifted windows lose at least one match by the preceding classification. Taking the maximum over all window starts proves the equivalence.
\end{proof}

The construction has $m=O(d^2)$, an alphabet of size $O(d)$, and $O(\log d)$-bit symbol codes. The proof does not establish hardness for an independently fixed finite temporal band: that restriction would require a new legality argument for the matches.

\subsection{EDR}
\label{ds:edr}

EDR assigns zero cost to a tolerance match and unit cost to a mismatching substitution, insertion, or deletion \citep{chen2005robust}. On our integer slice it is unit-cost Levenshtein distance, written $E$. This is the raw recurrence on the encoded input, independently of normalization used in applications. Use fresh, coordinate-private alphabets and define
\begin{align*}
 P_k(0)&=(A_k,B_k),&P_k(1)&=(D_k,C_k),\\
 Q_k(0)&=(D_k,B_k),&Q_k(1)&=(A_k,E_k),\\
 \Phi_P(u)&=c_1P_1(u_1)\cdots c_dP_d(u_d),\\
 \Phi_Q(v)&=c_1Q_1(v_1)\cdots c_dQ_d(v_d).
\end{align*}
Here each $c_k$ is a new private anchor. Both cores have length $3d$. A nonconflicting local pair has one equal position and one mismatch; the conflicting pair has no common symbol. Its EDR is therefore $1+u_kv_k$: a single edit cannot turn one of the latter length-two strings into the other. Disjoint phase alphabets also give
\begin{equation}
 L(\Phi_P(u),\Phi_Q(v))=2d-\langle u,v\rangle.
 \label{ds:core-lcs}
\end{equation}
For equal-length strings of length $\ell$, one has
\begin{equation}
 E(A,B)\ge\ell-L(A,B).
 \label{ds:edit-lcs}
\end{equation}
Indeed, a script of cost $t$ deletes or substitutes at most $t$ source symbols; its remaining source symbols give a common subsequence of length at least $\ell-t$. Equation~\eqref{ds:core-lcs} and a diagonal substitution script now prove the exact identity
\begin{equation}
 E(\Phi_P(u),\Phi_Q(v))=d+\langle u,v\rangle.
 \label{ds:core-edit}
\end{equation}

Let $Z$ be a fresh buffer symbol, $M=2d$, and set
\begin{align*}
 Q(v)&=Z^M\Phi_Q(v)Z^M,\qquad m=7d,\\
 S&=Z^M\Phi_P(u^{(1)})Z^{2M}\Phi_P(u^{(2)})\cdots\\
  &\hspace{21mm}\cdots Z^{2M}\Phi_P(u^{(n)})Z^M.
\end{align*}
The threshold is $\tau_E=d$. A canonical window has $M$ buffer symbols on either side of one complete core. Consecutive core starts are separated by $7d=m$, so $|S|=nm$.

\begin{lemma}[Full-core projection]
\label{ds:projection}
If a length-$m$ window $X$ contains a complete data core $\Phi_P(u)$, then all its other symbols are $Z$, and
\[
 E(X,Q(v))\ge E(\Phi_P(u),\Phi_Q(v)).
\]
\end{lemma}
\begin{proof}
Consecutive data cores are separated by $2M=4d$ buffer symbols. A window containing one complete core of length $3d$ and even one symbol of the next or previous core would span at least $3d+4d+1>m=7d$ positions. Thus all positions outside the contained core are buffer symbols.

Fix any global edit alignment between this window and the query. Regard it as an ordered list of columns, each consuming one symbol on at least one side. Transform each column as follows. Preserve a column that consumes two core symbols. If it pairs one core symbol with a buffer symbol, replace the buffer symbol by a gap. Preserve a column that already pairs a core symbol with a gap. Delete every column with no core symbol.

Every core symbol of each side is still consumed exactly once, and the relative order on each side is unchanged. No retained column has two gaps. The resulting list is therefore a legal global edit alignment between the two cores, even if many buffer symbols were matched to core symbols in the original alignment. Since the buffer symbol is absent from both core alphabets, a core--buffer column is a mismatch of cost one; its replacement is a gap column of the same cost. Preserved columns keep their cost, and deleted columns have nonnegative cost. The projection never increases total cost. Each original column is processed once, so an original substitution cannot be charged twice. This proves the inequality for every full alignment, and minimizing over such alignments yields the claimed distance bound.
\end{proof}

\begin{lemma}[Partial-core windows]
\label{ds:partial}
If a length-$m$ window contains no complete core, then
\begin{align}
 L(X,Q(v))&\le\max\{2M,M+3d-1\}=5d-1,
 \label{ds:partial-lcs}\\
 E(X,Q(v))&\ge 2d+1>d. \nonumber
\end{align}
\end{lemma}
\begin{proof}
Such a window has the form $A Z^{2M} B$, where $A$ and $B$ are proper suffix and prefix portions of adjacent cores, with $|A|+|B|=3d$. Their core phase sequence is cyclically shifted; the underlying vectors may differ. In particular, $A B$ cannot equal the query core: its first symbol is not the private initial anchor $c_1$. Since both have length $3d$, they have at most $3d-1$ core matches. This also follows by separating the two sides of the phase cut, as in Lemma~\ref{ds:lcss-separation}.

A common subsequence using no core symbol has at most $2M$ buffer matches. If it uses a core symbol in $A$, the central buffer can only match the query's right buffer; if it uses a core symbol in $B$, it can only match the left buffer. If it uses both sides, the central buffer lies between two core matches and cannot match either query buffer. Thus any subsequence with a core match has at most $M$ buffer matches and $3d-1$ core matches. These are alternative cases, giving the maximum in \eqref{ds:partial-lcs}, not the sum of both buffer allowances. Equation~\eqref{ds:edit-lcs} with $m=7d$ proves the EDR bound.
\end{proof}

\begin{theorem}[EDR all-window certificate]
\label{ds:edr-separation}
The construction satisfies
\begin{align*}
 &\exists u\in U:\ \langle u,v\rangle=0\\
 &\hspace{12mm}\Longleftrightarrow
 \min_{X\in\operatorname{Win}_m(S)}E(X,Q(v))\le d.
\end{align*}
\end{theorem}
\begin{proof}
An orthogonal canonical window has cost $d$ by \eqref{ds:core-edit}, matching its buffers for free. In a No instance, a full-core window costs at least $d+1$ by Lemma~\ref{ds:projection}, and any other window costs at least $2d+1$ by Lemma~\ref{ds:partial}.
\end{proof}

\subsection{Swale}
\label{ds:swale}

We use the global threshold-match recurrence of \citet{morse2007efficient}, with fixed reward $+r$ and single-sided gap penalty $-r$, for an exactly representable fixed $r>0$. There is \emph{no mismatching diagonal transition}. With $W_{i,0}=-ir$ and $W_{0,j}=-jr$, the recurrence is
\[
 W_{i,j}=\begin{cases}
 W_{i-1,j-1}+r,&a_i=b_j,\\
 \max\{W_{i-1,j},W_{i,j-1}\}-r,&a_i\ne b_j.
 \end{cases}
\]
Let $W(A,B)=W_{|A|,|B|}$. Equality here implements the fixed $\epsilon\in[0,1)$ tolerance slice.

\begin{lemma}[Swale--LCSS identity]
\label{ds:swale-identity}
For arbitrary lengths $a=|A|$ and $b=|B|$,
\[
 W(A,B)=\bigl(3L(A,B)-a-b\bigr)r.
\]
\end{lemma}
\begin{proof}
For each pair of prefixes, put $T_{i,j}=(3L(A[1..i],B[1..j])-i-j)r$. The boundary values agree with those of $W$. At equal last symbols, the LCSS value is the previous diagonal value plus one, so $T$ increases by $r$. At unequal symbols, the LCSS value is the maximum of the two neighboring values, and the corresponding $T$ recurrence subtracts $r$. Induction on $i+j$ proves the identity.
\end{proof}

Reuse the EDR cores and double-buffer target and query, with $m=7d$, $M=2d$, and threshold $\tau_W=4dr$.

\begin{theorem}[Swale all-window separation]\label{ds:swale-separation}
A full-core window associated with $u$ satisfies
\[
 L(X,Q(v))\le2M+2d-\langle u,v\rangle,
\]
with equality on a canonical window. Every partial-core window has $W(X,Q(v))\le(d-3)r$. Consequently an orthogonal canonical window has score $4dr$, whereas every window in a No instance has score at most $(4d-3)r$, and
\begin{align*}
 &\exists u\in U:\langle u,v\rangle=0\\
 &\hspace{8mm}\Longleftrightarrow
 \max_{X\in\operatorname{Win}_m(S)}W(X,Q(v))\ge4dr.
\end{align*}
\end{theorem}
\begin{proof}[Proof sketch]
The window classification is exactly that used for EDR. In a full-core window at most $2M$ buffer symbols match. Removing those matches leaves a core common subsequence, bounded by~\eqref{ds:core-lcs}; canonical windows attain both allowances. Substitution in Lemma~\ref{ds:swale-identity} gives canonical score $(4d-3\langle u,v\rangle)r$ and the same upper bound for other full-core windows. For a partial-core window, Lemma~\ref{ds:partial} gives $L\le5d-1$, hence
\[
 W\le\bigl(3(5d-1)-14d\bigr)r=(d-3)r.
\]
These cases establish the two threshold directions. The EDR edit projection is not being identified with a Swale alignment: the transferred object is the already proved core-subsequence and buffer bound.
\end{proof}
Other reward/gap ratios, additional mismatch transitions, or time bands are not included in this certificate.

\subsection{TS-SW}
\label{ds:tssw}

We define a real-valued time-series score inspired by the local restart recurrence of \citet{smith1981identification}. The original paper supplies the local-alignment mechanism, not this particular numeric scoring function or its hardness result. With zero boundary values, let
\begin{align*}
 H_{i,j}=\max\bigl\{&0,\ H_{i-1,j-1}+4-|x_i-q_j|,\\
                   &H_{i-1,j}-1,\\
                   &H_{i,j-1}-1\bigr\},
\end{align*}
and let $T(X,Q)=\max_{i,j}H_{i,j}$. Both the reward constant $4$ and the unit gap penalty are fixed independently of $d$.

For clarity, the numeric gadgets used here are
\begin{align*}
 p(0)&=(0,0),&p(1)&=(2,0),\\
 q(0)&=(1,0),&q(1)&=(0,1).
\end{align*}
Concatenate them coordinate-wise and divide every coordinate by $d$ to obtain $\widetilde P(u),\widetilde Q(v)\in[0,2/d]^{2d}$. The four local cases give
\[
 \|\widetilde P(u)-\widetilde Q(v)\|_1
 =1+2\langle u,v\rangle/d.
\]
Let $c_j=8j$ and, for a vector $z\in\mathbb R^{2d}$, define
\[
 B(z)=(c_1,z_1,\ldots,c_{2d},z_{2d},c_{2d+1}).
\]
Concatenate the data blocks $B(\widetilde P(u))$; the query is $B(\widetilde Q(v))$. Set $m=4d+1$ and $\tau_T=4m-1$.

\begin{lemma}[Full-consumption rigidity]
\label{ds:tssw-rigidity}
For two length-$m$ sequences and any scoring rule $B-|x-y|$ with $B>1$ and unit gap penalties, a local traceback of score at least $mB-1$ consumes all positions by exactly $m$ diagonal events and has no gap event.
\end{lemma}
\begin{proof}
Follow the traceback from its zero restart state in the forward direction. Let $K$ count its diagonal consumption events, \emph{including the event entering the first matched cell}, and let $G$ count horizontal and vertical gap events. Its score is at most $KB-G\le KB$, and monotonicity gives $K\le m$. If $K\le m-1$, the score is at most $(m-1)B<mB-1$. Thus $K=m$, which already consumes all $m$ positions on both sides. No gap, omitted prefix, or omitted suffix is possible. The matched cells are exactly $(1,1),\ldots,(m,m)$. There are $m$ consumption events but only $m-1$ moves between successive matched cells.
\end{proof}

\begin{theorem}[TS-SW all-window certificate]
\label{ds:tssw-separation}
For the fixed score defined above,
\begin{align*}
 &\exists u\in U:\ \langle u,v\rangle=0\\
 &\hspace{12mm}\Longleftrightarrow
 \max_{X\in\operatorname{Win}_m(S)}T(X,Q(v))\ge\tau_T.
\end{align*}
Every No-instance window has score at most $\tau_T-2/d$.
\end{theorem}
\begin{proof}[Proof sketch]
The canonical full diagonal has score $4m-1-2\langle u,v\rangle/d$. Each of its increments is positive, so an orthogonal canonical block is a valid local Yes witness. By Lemma~\ref{ds:tssw-rigidity}, any other traceback either is the full diagonal or uses at most $m-1$ diagonal consumption events and scores at most $4m-4=\tau_T-3$. A nonorthogonal canonical full diagonal loses at least $2/d$. On a shifted window its first difference is at least $8-2=6$, so even the full diagonal scores at most $4m-6=\tau_T-5$. These are all possibilities, and $2/d\le2$ gives the stated uniform No gap.
\end{proof}

The input coordinates have denominator $d$ and magnitude $O(d)$; they and the threshold have $O(\log d)$-bit encodings. The strictly positive gap $2/d$ suffices for exact rational comparison. It tends to zero and yields no fixed positive additive-error or fixed-factor approximation lower bound by this argument.

\begin{corollary}[Conditional indexing lower bounds]
\label{ds:indexing}
Fix the stated exactly representable tolerance and, for Swale, reward parameters. For each $\beta\ge1$ and $\delta\in(0,1)$, Deterministic SETH excludes a uniform exact static index with preprocessing $O(N^\beta)$ and worst-case query time $O(N^{1-\delta})$ for each of these four slices. Under Randomized SETH the same holds with any fixed per-query two-sided error $p<1/2$, jointly over preprocessing and query randomness. The query-length domain must cover the hard family generated by the dimension coefficient chosen after $(\beta,\delta)$ and, in the randomized case, $p$.
\end{corollary}
\begin{proof}[Proof sketch]
Invoke Lemma~\ref{lem:rt-transfer} with the all-window bounds of Lemma~\ref{ds:lcss-separation} and Theorems~\ref{ds:edr-separation}, \ref{ds:swale-separation}, and~\ref{ds:tssw-separation}. The query lengths are $2d^2+4d$, $7d$, $7d$, and $4d+1$, respectively; target lengths equal $n$ times these lengths. Construction takes $nd^{O(1)}$ and $d^{O(1)}$ time. Integer symbol codes, or the TS-SW rationals of denominator $d$, use $O(\log d)$ bits. For Swale a fixed exactly implementable reward multiplies integer counts. For a maximizing score, Proposition~\ref{fw:elementary} reverses the threshold polarity. Deterministic encoders preserve the error event, so the common transfer supplies both hypotheses without repeating their exponent calculation.
\end{proof}

\section{Soft-DTW}
\label{sec:rt-soft}

Fix a positive rational temperature $\gamma=a/b$, with fixed positive integers $a,b$, and a fixed band $W$ or the unconstrained path set $\Pi_W$. We study the following absolute-loss specialization of the soft-min path framework \cite{cuturi2017soft,sakoe1978dynamic}:
\begin{align}
C_\pi(X,Q)&=\sum_{(i,j)\in\pi}|x_i-q_j|,\\
F_\gamma(X,Q)&=-\gamma\ln\sum_{\pi\in\Pi_W}
                 \exp(-C_\pi(X,Q)/\gamma).
\label{rs:soft-definition}
\end{align}
Here $\ln$ is the natural logarithm, paired with the natural exponential. This is the raw soft-DTW score, which can be negative; it is not the bias-corrected soft-DTW divergence \cite{blondel2021differentiable}. Absolute local loss is not everywhere differentiable, and no differentiability, nonnegativity, or metric axiom is needed below.

\begin{lemma}[Entropy bound]\label{rs:entropy}
For any finite nonempty path family of size $H$ and $C_*=\min_\pi C_\pi$,
\[
C_*-\gamma\ln H\le F_\gamma\le C_*.
\]
In particular, the length-$m$ grid has $H\le3^{2m}$, also after a band restriction that retains its diagonal.
\end{lemma}
\begin{proof}[Proof sketch]
Apply Lemma~\ref{fw:entropy} with $M=H$. A length-$m$ monotone grid path uses at most $2m-2$ steps from three possible step types, so the number of paths is at most $3^{2m}$. A band that retains the diagonal keeps the path family nonempty and can only decrease this count.
\end{proof}
More generally, if minimum witness costs are at most $A$ in Yes instances and at least $B$ in No instances, a strict soft-score separation follows from $B-A>\gamma\ln H$. Uniformly scaling all witness costs by $s>0$ replaces this sufficient condition by $s(B-A)>\gamma\ln H$. The actual construction below also increases anchor separation, so that every noncanonical witness satisfies the required cost bound.

\begin{lemma}[Integer construction and all-window soft gap]
\label{rs:soft-gap}
Let $m=4d+1$ and set
\begin{align}
K&=4a+1,&\kappa&=Km+2,\nonumber\\
C&=\kappa(d+4)+Km+1,&\tau&=\kappa d.
\label{rs:soft-constants}
\end{align}
Use blocks $B_C(\kappa P(u))$ and query $B_C(\kappa R(v))$. If some $u\in A$ is orthogonal to $v$, a window has $F_\gamma\le\tau$. Otherwise every window has $F_\gamma>\tau+\kappa$. The data and query constructions take $O(nd)$ and $O(d)$ time, respectively, and use polynomial-in-$d$ integers.
\end{lemma}
\begin{proof}
Since $\ln3<2$ and $b\ge1$, the fixed integer $K$ satisfies $K>2\gamma\ln3$ without any transcendental computation in the construction. Hence, for $E=\gamma\ln H$,
\[
E\le2\gamma m\ln3<Km=\kappa-2.
\]
An orthogonal canonical diagonal has cost $\kappa d=\tau$, so Lemma~\ref{rs:entropy} gives the Yes bound. In a No canonical block, the diagonal costs at least $\kappa(d+2)$. All other canonical paths, and every path in a shifted window, contain a local difference at least
\[
C-2\kappa=\kappa(d+2)+Km+1>\kappa(d+2)
\]
by Lemma~\ref{lem:rt-barrier}, now with $\Delta=2\kappa$. Every No window therefore satisfies
\[
F_\gamma\ge\kappa(d+2)-E
>\tau+\kappa+2>\tau+\kappa.
\]
Only the data gadgets are multiplied by $\kappa$; the new common anchors have the values $jC$ prescribed above. Here $\kappa=O_\gamma(d)$, $C=O_\gamma(d^2)$, the largest symbol is $O_\gamma(d^3)$, and any integer path cost is $O_\gamma(d^4)$. All these discrete quantities use $O_\gamma(\log d)$ bits. None of the construction steps evaluates $\exp$, $\ln$, or $F_\gamma$.
\end{proof}

\begin{theorem}[Exact soft-DTW predicates in Word-RAM]\label{thm:rt-soft}
For fixed positive rational $\gamma$ and fixed $W$ (or no band), consider a uniform Word-RAM index that, on finitely encoded integer sequences and an integer threshold, returns the exact mathematical truth value of $\min_X F_\gamma(X,Q)\le\tau$. Deterministic SETH excludes preprocessing $O(N^\beta)$ together with worst-case query time $O(N^{1-\delta})$, for any fixed $\beta\ge1$ and $\delta\in(0,1)$. Under Randomized SETH the same holds for fixed per-input error probability $p<1/2$. The query-length domain must contain the hard lengths selected after these constants.
\end{theorem}
\begin{proof}[Proof sketch]
Lemma~\ref{rs:soft-gap} gives a deterministic all-window reduction with integer input and threshold, $N=n(4d+1)$, and construction times $O(nd)$ and $O(d)$. Apply Lemma~\ref{lem:rt-transfer}, with its adjusted OnlineOV exponents $(\beta+1/2,\delta/2)$. The encoder never evaluates $\exp$, $\ln$, or the soft score; it invokes the hypothetical Word-RAM Boolean predicate once. Thus the composed procedure stays in Word-RAM and preserves the per-input error event. The soft-score separation, rather than an assumed transcendental primitive, is what establishes its correctness.
\end{proof}
The theorem supplies neither exact transcendental Word-RAM primitives nor an algorithm for evaluating $F_\gamma$. A mathematical real-valued formula can define a predicate on finite inputs without changing the machine model of a hypothetical algorithm solving that predicate.

\paragraph{Extended Real-RAM statement.}\label{rs:extended-model}
For completeness, consider a separate model: one uniform finite program with a finite set of built-in real constants independent of input size, unit-cost exact arithmetic, comparison, $\exp$, and $\ln$. Inputs remain finitely encoded, and an infinite table depending on $n$, $c$, or the input is forbidden. Assume directly that in this model, for every fixed preprocessing exponent $b_0\ge1$ and saving $\epsilon\in(0,1)$, there is $c_0>0$ such that for every fixed $c\ge c_0$, OnlineOV on $d=\lceil c\log n\rceil$ is incompatible with preprocessing $O(n^{b_0})$ and query time $O(n^{1-\epsilon})$. The time guarantee is uniform over all sufficiently large $n$. For randomized indexes require the corresponding branch for every fixed per-input error $p<1/2$, allowing $c$ to depend on $p$.
\begin{corollary}[Model-relative soft-DTW lower bound]
\label{rs:extended-soft}
Under this direct extended-Real-RAM OnlineOV hypothesis, the time exclusions in Theorem~\ref{thm:rt-soft} also hold for indexes in that extended model, with the matching deterministic or randomized branch.
\end{corollary}
\begin{proof}[Proof sketch]
Use the same integer encoders and all-window gap from Lemma~\ref{rs:soft-gap}. The resource calculation in Lemma~\ref{lem:rt-transfer} now contradicts the stated OnlineOV hypothesis in the extended model, with $(b_0,\epsilon)=(\beta+1/2,\delta/2)$. This changes the assumed source lower bound, not the numerical construction or its error event.
\end{proof}
Ordinary SETH is not asserted to imply this stronger machine-model hypothesis. Rejecting it removes this extended-model conclusion, not the Word-RAM Boolean-predicate theorem.

\paragraph{Finite-precision scores.}\label{rs:finite-soft}
If a score-producing implementation certifies $|\widehat F_\gamma-F_\gamma|\le\kappa/2$ uniformly on every window of every constructed instance, then the threshold $\tau+\kappa/2$ separates the estimated Yes and No scores: the former is at most that threshold and the latter is strictly greater. Taking a minimum over all windows preserves the same error bound. More generally, Yes scores at most $A$ and No scores strictly above $B$ are separated by the midpoint whenever the uniform error is at most $(B-A)/2$. This conditional numerical transfer uses the complexity assumption appropriate to the implementation's machine model; it does not certify a specific floating-point implementation.

An exact log-sum-exp rearrangement is
\[
F_\gamma=C_*-\gamma\ln
\sum_\pi\exp\bigl(-(C_\pi-C_*)/\gamma\bigr).
\]
Its largest exponential is $1$, which prevents all terms from underflowing together because of a large common cost. Terms with very large excess costs may still underflow, and summation and rounding errors still require a uniform bound. A black-box exact Boolean predicate and a rounded score with a proved error guarantee are different promises; the former needs no numerical error certificate for the reduction in Theorem~\ref{thm:rt-soft}.

\section{Effects of normalization}
\label{fw:normalization}

For a nonconstant length-$m$ real sequence, write $\mu_X=m^{-1}\sum_i x_i$, $\sigma_X^2=m^{-1}\sum_i(x_i-\mu_X)^2$, and $Z(X)=(X-\mu_X\mathbf1)/\sigma_X$. The following boundary statements concern explicit hypotheses on a construction, not a general impossibility of normalized reductions.

\begin{proposition}[Affine invariance and common moments]
\label{fw:moments}
For $a>0$, $Z(aX+b\mathbf1)=Z(X)$. If a query and every relevant window share mean $\mu_0$ and positive standard deviation $\sigma_0$, and a score satisfies
\[
D(aX+b\mathbf1,aQ+b\mathbf1)=a^rD(X,Q),
\]
then $D(Z(X),Z(Q))=\sigma_0^{-r}D(X,Q)$. Thresholds and additive gaps scale by this factor whenever the common statistics and comparisons are legal in the declared model.
\end{proposition}
\begin{proof}
Mean and standard deviation transform as $a\mu_X+b$ and $a\sigma_X$. For the second statement apply the homogeneity identity with $a=1/\sigma_0$ and $b=-\mu_0/\sigma_0$.
\end{proof}
For $a<0$ the normalized sequence instead changes sign.  Multiplying only a subset of coordinates is not a common affine transformation. More importantly, common moments of canonical blocks do not establish common moments for every sliding window.

\begin{proposition}[Common window means and periodicity]\label{fw:periodicity}
For $1\le m\le N$, all length-$m$ windows of $S=(s_1,\ldots,s_N)$ have the same mean if and only if
\[
 s_{i+m}=s_i\qquad(1\le i\le N-m).
\]
Under this condition every window has the same contents as one of the first $K=\min\{m,N-m+1\}$ windows.
\end{proposition}
\begin{proof}
If $W_i=\sum_{j=0}^{m-1}s_{i+j}$, adjacent sums satisfy $W_{i+1}-W_i=s_{i+m}-s_i$. Equality of all window means is therefore equivalent to the displayed periodicity condition. For any start $t$, choose $1\le r\le m$ with $t\equiv r\pmod m$. Since $r\le t\le N-m+1$, periodicity identifies the entire window at $t$ with the one at $r$. This also proves the stated bound near the end of the target.
\end{proof}

\begin{corollary}[A static index for the common-mean promise]\label{fw:periodic-index}
Assume target symbols admit exact $O(1)$-time comparison, and that a score depends only on window contents, the query, and explicit parameters. If one window of length $m$ is scored in $T(m)$ time, the common-mean promise admits $O(N^2)$ preprocessing, $O(N)$ space, and query time
\[
 O\bigl(m^2+mT(m)\bigr).
\]
Preprocessing does not need to know the query length.
\end{corollary}
\begin{proof}[Proof sketch]
Store $S$ and, for every $m\le N$, the flag $\mathrm{valid}[m]=\bigwedge_{i=1}^{N-m}(s_i=s_{i+m})$. The total number of comparisons is $N(N-1)/2$. A query in the promise domain passes its flag check and needs only the first $K$ windows from Proposition~\ref{fw:periodicity}. Copying these windows takes $O(Km)$ time and scoring them takes $O(KT(m))$, with $K\le m$. Store only the target and the $N$ flags, giving linear space.
\end{proof}

For $m=\operatorname{polylog}N$ and $T(m)=\operatorname{poly}(m)$ this is a polylogarithmic query bound. Consequently, forcing a common mean on \emph{every} window cannot establish a near-linear query exclusion for every polynomial preprocessing exponent. Queries outside this promise may use a separate scan. A score depending on absolute window positions needs a separate argument; local timestamps reset within the window are part of its contents. Positive variance is still required when the score uses $Z$-normalization. The common-moment identity in Proposition~\ref{fw:moments} remains valid, but our positive normalized construction balances only canonical blocks and bounds shifted-window perturbations separately.

\begin{lemma}[Perturbation under a shared anchor skeleton]
\label{fw:attenuation}
Let $X,Y$ have the same length $m$ and anchor positions and values, and differ only on bounded data coordinates with bounded differences. Suppose, uniformly in the family, that
\[
\sigma_X,\sigma_Y=\Theta(A),\quad
|\mu_X-\mu_Y|+|\sigma_X-\sigma_Y|=O(1),
\]
and all centered coordinates have magnitude $O(A)$, where $A\to\infty$. Then
\begin{align*}
\|Z(X)-Z(Y)\|_\infty&=O(A^{-1}),\\
\|Z(X)-Z(Y)\|_1&=O(m/A).
\end{align*}
\end{lemma}
\begin{proof}
At coordinate $i$, the difference is
\begin{align*}
Z(X)_i-Z(Y)_i
={}&\frac{x_i-y_i-(\mu_X-\mu_Y)}{\sigma_X}\\
 &+(y_i-\mu_Y)\frac{\sigma_Y-\sigma_X}{\sigma_X\sigma_Y}.
\end{align*}
Both terms are $O(1/A)$ by the hypotheses.  Take the maximum or sum over all coordinates.  In particular, normalization can perturb anchor coordinates even when their original values agree.
\end{proof}

\begin{theorem}[Conditional attenuation of the data signal]
\label{fw:signal}
Assume Lemma~\ref{fw:attenuation} holds for the Yes/No constructions being compared with the same query.  Suppose the normalized score is $L(d)$-Lipschitz in its first input in $\ell_1$ on this family. Their score difference is at most $O(L(d)m(d)/A(d))$. If this quantity tends to zero, the bounded data perturbation alone cannot generate a fixed positive additive gap. For a pointwise additive score $F(X,Q)=\sum_i f(x_i,q_i)$, it suffices that $f(\cdot,q)$ is uniformly $\ell(d)$-Lipschitz; then $L(d)=\ell(d)$.
\end{theorem}
\begin{proof}
Apply the Lipschitz bound to the $\ell_1$ conclusion of Lemma~\ref{fw:attenuation}.  For the additive instance, $|F(X,Q)-F(Y,Q)|\le\sum_i\ell(d)|x_i-y_i|$.
\end{proof}
For the shared skeleton containing a constant fraction of evenly spaced anchors $c_j=jC$, with $m=\Theta(d)$ and bounded data values, the centered scale is $A=\Theta(dC)$.  Indeed, the anchor second moment is $O((dC)^2)$, while two constant fractions of anchors are separated by $\Theta(dC)$, giving a matching variance lower bound. For corresponding bounded data perturbations, the mean changes by $O(1)$; centering is an orthogonal projection, so the standard deviation changes by at most $\|X-Y\|_2/\sqrt m=O(1)$. Hence a pointwise score with $\ell(d)=O(1)$ has signal at most $O(1/C)$. These stability and Lipschitz conditions are not asserted for every elastic dynamic program or every shifted window.  Path repetition may increase $L(d)$, and a new normalized encoding may carry another signal.  Positive shrinking gaps still support exact predicates by Proposition~\ref{fw:elementary}; a specified finite-precision score needs its own gap-preservation argument.

\begin{proposition}[Perturbation of absolute-cost path scores]\label{fw:path-perturbation}
Let $X,Y,Q$ have length $m$ and $h=\|X-Y\|_\infty$. Fix the same nonempty, value-independent set $\Pi$ of monotone endpoint-to-endpoint paths with steps $(1,0),(0,1),(1,1)$. Define
\begin{align*}
 C_\pi(X,Q)&=\sum_{(i,j)\in\pi}|x_i-q_j|,\\
 D_\Pi(X,Q)&=\min_{\pi\in\Pi}C_\pi(X,Q),\\
 S_{\Pi,\gamma}(X,Q)&=-\gamma\ln\sum_{\pi\in\Pi}e^{-C_\pi(X,Q)/\gamma},\\
 F_\Pi(X,Q)&=\min_{\pi\in\Pi}\max_{(i,j)\in\pi}|x_i-q_j|,
\end{align*}
where $\gamma>0$ is fixed. Then
\begin{align*}
 |D_\Pi(X,Q)-D_\Pi(Y,Q)|&\le(2m-1)h,\\
 |S_{\Pi,\gamma}(X,Q)-S_{\Pi,\gamma}(Y,Q)|&\le(2m-1)h,\\
 |F_\Pi(X,Q)-F_\Pi(Y,Q)|&\le h.
\end{align*}
\end{proposition}
\begin{proof}
For each cell, the reverse triangle inequality bounds the change in absolute cost by $h$. A monotone path visits at most $2m-1$ cells. With $\varepsilon=(2m-1)h$, every path cost therefore changes by at most $\varepsilon$. Taking the minimum preserves this bound, proving the DTW assertion.

For soft-DTW, let $Z_X=\sum_{\pi\in\Pi}e^{-C_\pi(X,Q)/\gamma}$ and define $Z_Y$ analogously. The same pathwise inequalities give
\[
 e^{-\varepsilon/\gamma}Z_Y\le Z_X
       \le e^{\varepsilon/\gamma}Z_Y.
\]
Applying $-\gamma\ln$ proves the second bound. For DFD the maximum along one path changes by at most $h$; minimizing over the unchanged path family proves the third. Neither the optimal path nor the minimizing witness needs to remain the same after perturbation.
\end{proof}

\begin{corollary}[Signal bounds for a normalized anchor skeleton]\label{fw:path-signal}
Under all the hypotheses of Lemma~\ref{fw:attenuation}, compare $Z(X)$ and $Z(Y)$ with the same normalized query and the same value-independent path family. Absolute-cost DTW and raw soft-DTW change by $O(m/A)$; the DFD bottleneck changes by $O(1/A)$. For $m=\Theta(d)$ and the arithmetic anchor skeleton with $A=\Theta(dC)$, these bounds become $O(1/C)$ and $O(1/(dC))$, respectively.
\end{corollary}
\begin{proof}[Proof sketch]
Lemma~\ref{fw:attenuation} gives $\|Z(X)-Z(Y)\|_\infty=O(1/A)$. Substitute this bound into Proposition~\ref{fw:path-perturbation}. The sum and soft-min bounds use the path-length factor; the bottleneck bound does not.
\end{proof}

These are conditional bounds for the compared sequences, not a claim that every shifted window has the required moments. The path family must be unchanged by the perturbation, and TWED or MSM requires its own analysis. The soft-DTW estimate is an identity about real scores; it grants no exact transcendental Word-RAM primitive. A positive gap that tends to zero can still support an exact threshold reduction. The next section supplies such a construction with independently verified moments and finite-word comparisons.

\section{Normalized DTW and DFD}
\label{sec:z}

Raw-score reductions do not automatically survive independent window normalization.  We balance the moments of canonical blocks and bound, rather than eliminate, their changes in every shifted window.  This section concerns \emph{squared} local loss for DTW; it does not substitute squared loss for the absolute loss used elsewhere in the paper.

\paragraph{Scores and input domain.}
For an integer sequence $X=(x_1,\ldots,x_m)$ of positive population variance, put
\begin{gather*}
 \mu_X=m^{-1}\sum_i x_i,\qquad
 v_X=m^{-1}\sum_i x_i^2-\mu_X^2,\\
 Z_X(i)=\frac{x_i-\mu_X}{\sqrt{v_X}}.
\end{gather*}
Let $\Pi_m$ consist of the global monotone paths from $(1,1)$ to $(m,m)$ with steps $(1,0),(0,1),(1,1)$, either unrestricted or restricted to any fixed nonnegative Sakoe--Chiba band.  The diagonal is legal in either case. Define
\begin{align*}
 D_Z^{\mathrm{sqDTW}}(X,Q)
   &=\min_{\pi\in\Pi_m}\sum_{(i,j)\in\pi}
          (Z_X(i)-Z_Q(j))^2,\\
 D_Z^{\mathrm{sqDFD}}(X,Q)
   &=\min_{\pi\in\Pi_m}\max_{(i,j)\in\pi}
          (Z_X(i)-Z_Q(j))^2.
\end{align*}
Ordinary normalized discrete Fr\'echet distance is $D_Z^{\mathrm{DFD}}=\sqrt{D_Z^{\mathrm{sqDFD}}}$. There is no division by path length or outer square root in $D_Z^{\mathrm{sqDTW}}$.  The index receives raw integers and a rational threshold.  Each window and the query use their own moments; irrational normalized coordinates are not assumed to be input machine words. The legal instance domain requires positive variance for the query and every compared window.  Our construction satisfies this requirement. Write $\operatorname{Win}_m(S)$ for all contiguous length-$m$ windows of $S$.

\subsection{Balanced encoding}

For $a,b\in\{0,1\}$ use the six-coordinate gadgets
\[
\begin{array}{ll}
 P_0=(1,-1,0,0,0,0),&P_1=(0,0,1,-1,0,0),\\
 Q_0=(0,0,0,0,1,-1),&Q_1=(0,0,-1,1,0,0).
\end{array}
\]
Every gadget has coordinate sum zero and squared norm two.  Disjoint supports, except that $Q_1=-P_1$, give
\begin{equation}
 \|P_a-Q_b\|_2^2=4+4ab,\qquad
 \|P_a-Q_b\|_\infty=1+ab.
 \label{eq:z:gadget}
\end{equation}
For Boolean vectors $u,q\in\{0,1\}^d$, concatenate the respective gadgets to obtain $p(u),q^{\flat}(q)\in\{-1,0,1\}^{k}$, where
\[
 k=6d,\qquad C=d+8,\qquad m=2k+1=12d+1.
\]
Define the anchored block
\[
 B(x)=(C,x_1,2C,x_2,\ldots,kC,x_k,(k+1)C).
\]
For an ordered data family $A=(u^{(1)},\ldots,u^{(n)})$, allowing repetitions, set
\[
 S_A=B(p(u^{(1)}))\circ\cdots\circ B(p(u^{(n)})),
 \qquad Q_q=B(q^{\flat}(q)).
\]
Thus $N=nm$, without additional separators.  The $n$ complete blocks are the \emph{canonical} windows; the others are \emph{shifted} windows.

\begin{lemma}[Moments of every window]\label{lem:z:moments}
All canonical blocks and queries share sum $H$, squared sum $E$, mean $\mu=H/m$, and variance $v=E/m-\mu^2$, where
\begin{align*}
 H&=\frac{C(k+1)(k+2)}2,\\
 E&=\frac{C^2(k+1)(k+2)(2k+3)}6+2d,\\
 v&=\frac{C^2k(k+1)(k+2)(5k+7)}{12(2k+1)^2}
       +\frac{2d}{m}\ \ge\ \frac{C^2k^2}{12}>0.
\end{align*}
Every shifted window $X$ has sum $H+e$ and squared sum $E+f$, with $|e|\le1$, $|f|\le2$.  Consequently
\begin{gather*}
 \mu_X=\mu+e/m,\qquad v_X=v+\eta,\\
 \eta=f/m-2\mu e/m-e^2/m^2,
\end{gather*}
and $|\eta|<C$, $v_X>2v/3>0$.
\end{lemma}
\begin{proof}
The anchors are $C,2C,\ldots,(k+1)C$ and the $d$ gadgets each contribute sum zero and squared sum two.  Summing these terms gives $H,E$ and the displayed variance identity.  Its lower bound follows from
\begin{multline*}
 (k+1)(k+2)(5k+7)-k(2k+1)^2\\
   =k^3+18k^2+30k+14>0.
\end{multline*}
A shifted window is a proper suffix of one block followed by a proper prefix of the next.  Each physical phase occurs once, so its anchor multiset is unchanged.  A shift by $r\in\{1,\ldots,m-1\}$ replaces the first $\lfloor r/2\rfloor$ data coordinates of the left block with the corresponding prefix of the right block.  Complete gadgets have identical moments, leaving at most one partial gadget.  The seven prefix sums of $P_0,P_1$, including the empty prefix, are respectively
\[
 (0,1,0,0,0,0,0),\qquad (0,0,0,1,0,0,0),
\]
and their prefix squared sums are
\[
 (0,1,2,2,2,2,2),\qquad (0,0,0,1,2,2,2).
\]
Their differences give $|e|\le1$, $|f|\le2$, including shifts at anchor and gadget boundaries.  When $n=1$ there are no shifted windows. Expanding the variance gives $\eta$.  Since $k\ge6$, $C\ge1$, and $\mu\le C(k+2)/2$,
\begin{align*}
 |\eta|&\le 2/m+2\mu/m+1/m^2\\
 &<C\frac{k+2}{2k}+\frac1k+\frac1{4k^2}\\
 &\le C\left(\frac12+\frac2k+\frac1{4k^2}\right)<C.
\end{align*}
Thus $|\eta|/v<12/(Ck^2)\le1/3$, proving positivity and the bound on $v_X$.
\end{proof}

\begin{lemma}[Normalized first-cell barrier]\label{lem:z:first}
For every shifted window $X$,
\[
 |Z_X(1)-Z_{Q_q}(1)|>\frac{C-5}{\sqrt v}
                         =\frac{d+3}{\sqrt v}.
\]
\end{lemma}
\begin{proof}
Let $a$ be its first element and $\alpha=\sqrt{v/v_X}$. The element $a$ is a data coordinate in $[-1,1]$ or an anchor $jC$ with $j\ge2$; hence $|a-C|\ge C-1$ and $|a-\mu|\le(k+1)C$. Lemma~\ref{lem:z:moments}, and the derivative of $(1+t)^{-1/2}$ for $|t|\le1/3$, imply
\[
 \alpha<3/2,\qquad |\alpha-1|\le3C/(2v).
\]
Since the query starts with $C$,
\begin{multline*}
 \sqrt v\,(Z_X(1)-Z_{Q_q}(1))=(a-C)\\
   +(\alpha-1)(a-\mu)-\alpha e/m.
\end{multline*}
The last two terms have total absolute value strictly less than
\[
 \frac{18(k+1)}{k^2}+\frac3{2(2k+1)}
       \le\frac{47}{13}<4;
\]
both summands decrease for $k\ge6$.  Subtracting this error from $|a-C|\ge C-1$ proves the claim.
\end{proof}

\begin{theorem}[Exact identities over all windows]\label{thm:z:identity}
For $n,d\ge1$ let $h=\min_{u\in A}\langle u,q\rangle$.  Then
\begin{align*}
 \min_{X\in\operatorname{Win}_m(S_A)}D_Z^{\mathrm{sqDTW}}(X,Q_q)
   &=\frac{4d+4h}{v},\\
 \min_{X\in\operatorname{Win}_m(S_A)}D_Z^{\mathrm{sqDFD}}(X,Q_q)
   &=\begin{cases}1/v,&h=0,\\4/v,&h>0.\end{cases}
\end{align*}
Thus thresholds $4d/v$ and $1/v$, respectively, decide whether $A$ contains a vector orthogonal to $q$.
\end{theorem}
\begin{proof}
A canonical block and query share their moments.  On the diagonal the anchors cancel, so~\eqref{eq:z:gadget} gives squared DTW cost $(4d+4\langle u,q\rangle)/v$ and squared DFD cost $1/v$ or $4/v$. Any nondiagonal path has a first horizontal or vertical departure from the diagonal.  As anchors alternate with data, that cell aligns an anchor and a data coordinate and costs at least
\[
 (C-1)^2/v=(d+7)^2/v>8d/v.
\]
It exceeds every canonical diagonal cost for either score.  Hence the diagonal is optimal in every canonical block.  Every path in a shifted window visits $(1,1)$, whose cost, by Lemma~\ref{lem:z:first}, exceeds $(d+3)^2/v>8d/v$, since $(d+3)^2-8d=(d-1)^2+8>0$. Both path sum and bottleneck dominate this cell.  All shifted windows are therefore more expensive than all canonical candidates.  Minimizing the canonical costs proves the identities.
\end{proof}

\subsection{Exact comparisons and indexing}

\begin{lemma}[Exact comparison without real-root primitives]
\label{lem:z:arithmetic}
On the constructed instances, raw inputs, rational thresholds, and the integer coefficients needed for exact normalized squared-score dynamic programming have $O(\log d)$ bits.  One window can be evaluated in $O(m^2)$ Word-RAM operations, using a constant number of machine words per coefficient when $d=\Theta(\log n)$.
\end{lemma}
\begin{proof}
For a fixed window and query put
\[
 s_X=\sum_i x_i,\quad J_X=m\sum_i x_i^2-s_X^2,\quad a_i=mx_i-s_X,
\]
and define $s_Q,J_Q,b_j$ similarly.  Both $J$ values are positive, and $Z_X(i)=a_i/\sqrt{J_X}$.  With $D=J_XJ_Q$, a local squared cost is
\begin{equation}
 \frac{A_{ij}-B_{ij}\sqrt D}{D},\quad
 A_{ij}=a_i^2J_Q+b_j^2J_X,\quad B_{ij}=2a_ib_j.
 \label{eq:z:root}
\end{equation}
Every local cost for this window shares the same root.  Path summation adds the two integer coefficients; bottleneck maxima and DP minima select existing expressions.  No operation raises the algebraic degree. To compare expressions, determine the sign of $a-b\sqrt D$. For $b=0$, use the sign of $a$.  For $b>0,a\le0$ the sign is negative; for $b>0,a>0$ it is the sign of $a^2-b^2D$.  For $b<0,a\ge0$ the sign is positive; for $b<0,a<0$ it is the sign of $b^2D-a^2$. These branches cover equality, including square $D$, without squaring across an unchecked sign.  Comparison to a rational threshold first clears its positive denominator.  Thus it suffices to compare each window with the common threshold; values from different quadratic fields need not be compared directly.

Raw values are $O(d^2)$ and $m=O(d)$, so $|s_X|,|a_i|=O(d^3)$ and $J_X=O(d^6)$. The local coefficients are $A_{ij}=O(d^{12})$, $B_{ij}=O(d^6)$. A path visits at most $2m-1$ cells.  Its coefficients and all squared comparison products remain $d^{O(1)}$, hence have $O(\log d)$ bits. They fit a constant number of $\Theta(\log N)$-bit words on the hard family.  Standard DP uses $O(m^2)$ states and constant work per state. Finally $v=J_0/m^2$ for $J_0=mE-H^2$, and the thresholds are $4dm^2/J_0$ and $m^2/J_0$, also rational with $O(\log d)$-bit numerators and denominators.  The reduction generates raw integers and these rationals, without first evaluating any normalized coordinates.
\end{proof}

\begin{corollary}[Rational threshold for ordinary DFD]\label{cor:z:rational}
The same construction reduces orthogonality to ordinary normalized DFD with rational threshold $r=m/\lfloor\sqrt{J_0}\rfloor$.
\end{corollary}
\begin{proof}
The two optimal values are $m/\sqrt{J_0}$ and $2m/\sqrt{J_0}$. The variance bound gives $\sqrt{J_0}>2$.  Writing $b=\lfloor\sqrt{J_0}\rfloor$ therefore yields
\[
 m/\sqrt{J_0}\le m/b<2m/\sqrt{J_0}.
\]
Integer binary search using squared comparisons computes $b$ in $\operatorname{poly}(d)$ time with $O(\log d)$-bit integers.
\end{proof}

\begin{theorem}[Window-normalized indexing lower bounds]\label{thm:z:index}
Fix $\beta\ge1$ and $0<\delta<1$.  Under deterministic SETH, no uniform Word-RAM exact index for $D_Z^{\mathrm{sqDTW}}$, $D_Z^{\mathrm{sqDFD}}$, or $D_Z^{\mathrm{DFD}}$ has preprocessing $O(N^\beta)$ and worst-case query time $O(N^{1-\delta})$, when its length domain covers the hard family selected below.  Under randomized SETH the same exclusion holds for each fixed per-input two-sided error $0\le p<1/2$.  This probability includes preprocessing and query randomness, and time is worst-case over all random choices.
\end{theorem}
\begin{proof}[Proof sketch]
Apply the full transfer proof of Lemma~\ref{lem:rt-transfer}. For its selected fixed $c$, set $d=\lceil c\log_2 n\rceil$. The data encoder uses only $A$, takes $O(nd)$ time, and gives $N=n(12d+1)$. The query and squared thresholds take $O(d)$ time; the rational ordinary-DFD threshold takes $\operatorname{poly}(d)$ time. Theorem~\ref{thm:z:identity} gives all-window equivalence and Lemma~\ref{lem:z:arithmetic} supplies finite-word comparisons. Thus the same adjusted exponents $(\beta+1/2,\delta/2)$ contradict OnlineOV hardness. The required lengths are $12\lceil c\log_2 n\rceil+1$. Deterministic encoders preserve the index's joint per-input error event, so the randomized branch follows under randomized SETH.
\end{proof}

\begin{corollary}[Feasible-window approximation below two]\label{cor:z:approx}
For every fixed $1\le\rho<2$, Theorem~\ref{thm:z:index} also excludes an index returning a feasible window $\widehat X$ with
\[
 D_Z^{\mathrm{DFD}}(\widehat X,Q)
 \le\rho\min_{X\in\operatorname{Win}_m(S)}D_Z^{\mathrm{DFD}}(X,Q).
\]
The conclusion also applies when a window and a feasible coupling with that cost guarantee are returned.  In the randomized case feasibility and the guarantee must hold jointly with probability at least $1-p$ for each fixed input.
\end{corollary}
\begin{proof}[Proof sketch]
Reuse the feasible-output verification reduction of Corollary~\ref{rs:dfd-approx}. Choose a fixed rational $g$ with $\rho<g<2$, retain $S_A$, and reject invalid output. For a returned window, evaluate its exact squared DFD by Lemma~\ref{lem:z:arithmetic} and compare with $g^2/v$. A successful Yes output has squared score at most $\rho^2/v<g^2/v$; every feasible No output has squared score at least $4/v$. A returned feasible coupling with the promised bottleneck likewise upper-bounds its window's exact distance. The extra $O(m^2)=n^{o(1)}$ verification time is absorbed by the common transfer. Feasibility and approximation must hold jointly with probability at least $1-p$; this argument does not cover $\rho=2$, an expectation-only promise, or an arbitrary underestimating value.
\end{proof}

\subsection{Dyadic rounding}

The exact squared DTW gap $4/v$ decreases with $d$.  Nevertheless a prespecified precision growing only with query length preserves it. The following result concerns its stated score definition, rather than an unspecified floating-point implementation.

\begin{corollary}[Length-dependent dyadic rounding]\label{cor:z:rounding}
Let $R(m)$ be the least power of two not smaller than $8m^5$. Independently normalize each positive-variance window and query, replace each local squared cost $\ell_{ij}$ with $\lfloor R(m)\ell_{ij}\rfloor/R(m)$, and perform all subsequent DP additions, maxima and minima exactly on integers with common denominator $R(m)$.  Exact rational-threshold indexing for these rounded squared DTW and squared DFD scores obeys Theorem~\ref{thm:z:index}, under its respective deterministic or randomized assumption and hard-length coverage.
\end{corollary}
\begin{proof}
Use the same construction.  Since $C=d+8\le m$ and $(k+1)C\le m^2$, its common canonical variance satisfies $0<v\le m^4$. Every local rounding error is strictly smaller than
\[
 1/R\le 1/(8m^5)\le 1/(8mv).
\]
A path has at most $2m-1$ cells, so its sum loses less than $1/(4v)$. This also bounds the error after minimizing over paths: evaluate the exact cost on a rounded-optimal path for the lower bound, and use downward rounding for the upper bound.  The same reasoning applies to the finite minimum over all windows.  In particular, neither the optimal path nor the optimal window needs to remain unchanged. By Theorem~\ref{thm:z:identity}, exact squared DTW has yes value $4d/v$ and no value at least $(4d+4)/v$.  Therefore the rounded scores remain strictly separated by the rational threshold
\[
 T_D=(4d+2)/v.
\]
For squared DFD a path maximum loses less than $1/R$, and both levels of minimization preserve that error.  Its values $1/v$ and $4/v$ remain strictly separated by $T_F=5/(2v)$.

For any positive-variance length-$m$ sequence, $\sum_i Z_X(i)^2=m$, so $|Z_X(i)|\le\sqrt m$ and $0\le\ell_{ij}\le4m$.  A local quantized integer is at most $4mR$, and a path integer sum is less than $8m^2R$. Since $R<16m^5$, all these integers are polynomial in $m$. The exact floor does not require an approximate square root. Using~\eqref{eq:z:root}, testing an integer $t\le R\ell_{ij}$ means testing the sign of
\[
 (RA_{ij}-tD)-RB_{ij}\sqrt D.
\]
The sign branches in Lemma~\ref{lem:z:arithmetic} decide it exactly. Binary search between the accepted lower bound $0$ and rejected upper bound $4mR+1$ obtains $\lfloor R\ell_{ij}\rfloor$ in $O(\log m)$ comparisons.  All operands, including squared comparison products, are still $d^{O(1)}$ on the construction.  A window can therefore be evaluated in $O(m^2\log m)$ finite-word operations.  If its final DP integer is $L$, threshold comparisons can be written without roots as
\[
 LJ_0\le(4d+2)m^2R\quad\hbox{or}\quad 2LJ_0\le5m^2R.
\]
The rule $R(m)$ is fixed independently of the data and target index exponents; its construction uses $O(\log m)$ integer doublings. Data construction still costs $O(nd)$, and query construction costs $\operatorname{poly}(d)$.  The same adjusted exponents and deterministic reduction maps used in Theorem~\ref{thm:z:index} give the two SETH consequences.  Single-window evaluation certifies the score's finite-word semantics; the reduction does not enumerate all windows when answering a query.
\end{proof}

The precision is $O(\log m)$ fractional bits, or $O(\log\log N)$ on the hard family.  Corollary~\ref{cor:z:rounding} does not certify fixed IEEE precision, rounded normalized coordinates, or fresh rounding at each DP update.  Nor do these results cover normalized absolute-loss DTW, free endpoints, path-length normalization, or other elastic scores without separate arguments.  Raw $L_p$ subsequence indexing already has a direct predecessor~\cite{chao2022inherent}; the additional steps here are the balanced canonical moments, all-shift perturbation bounds, independently normalized first-cell barrier, and exact finite-word comparison.

\section{Discussion}
\label{sec:conclusion}

The reductions establish conditional indexing lower bounds for the measures and parameter restrictions in Table~\ref{in:catalog}. They construct a static target independently of the query and preserve the decision over all equal-length windows. The dimension coefficient is selected after the preprocessing and query exponents; constant query lengths and an arbitrary prescribed small logarithmic coefficient are outside these bounds. Deterministic and randomized indexes use their respective SETH assumptions.

The required alignment arguments depend on the measure. Anchors and mandatory endpoints suffice for several global path scores. Edit operations require additional cost bounds, derivative DTW requires control of the inverse recurrence and block boundaries, and subsequence and local scores require bounds on skipped or unconsumed positions. CID adds a global complexity factor, while soft-DTW adds the entropy of the path family.

Independent normalization admits exact lower bounds when canonical moments agree and shifted moments can be controlled. The resulting shrinking gap is consistent with the conditional attenuation bound in Section~\ref{fw:normalization}. Requiring every adjacent window to have the same mean would instead impose periodicity.

On that promise, Corollary~\ref{fw:periodic-index} gives an explicit quadratic-preprocessing, linear-space index. Under the separately stated moment-stability premises, Proposition~\ref{fw:path-perturbation} and Corollary~\ref{fw:path-signal} quantify the attenuation for absolute-cost DTW, soft-DTW and DFD. The squared-DTW and TS-SW gaps prove exact hardness; they do not give fixed-factor approximation hardness. For raw and normalized DFD, the gap excludes each fixed feasible approximation factor below two. The endpoint two is a limit of these constructions, not a general approximation upper bound.

The exact quadratic-field comparisons and the prescribed dyadic rule use finite-word arithmetic. Their guarantees concern those operations, and do not extend to arbitrary IEEE arithmetic, intermediate rounding, or other local losses. The cSoft-DTW Word-RAM reduction invokes an exact predicate; its extension to unit-cost exact transcendental operations assumes OnlineOV hardness in that model.

Approximation, restricted input distributions, other query-length domains, free endpoints, and other normalization rules remain separate questions. The worst-case bounds are compatible with practical gains from filtering and early abandonment.

\paragraph{Verification status.}
The results are supported by the paper proofs and proof sketches above. We do not claim an end-to-end machine-checked formalization. A reproducible Lean or Rocq/Coq development would need to supply the score-specific all-window semantics, the complete parameter-sensitive OnlineOV argument, the joint probability space for randomized preprocessing and queries, and the machine resource bounds. Checking local Boolean identities or a conditional dependency interface alone does not verify these remaining obligations. The normalized moment estimates, quadratic-root comparisons, and prescribed rounding rule are also part of those obligations.

\bibliographystyle{plainnat}
\bibliography{references}
\end{document}